\documentclass[a4paper,aps,floatfix,prl,twocolumn,nofootinbib,superscriptaddress]{revtex4-1}
\usepackage[utf8]{inputenc}
\usepackage[T1]{fontenc}

\usepackage{amsmath, amssymb, amsfonts, amsthm, thm-restate, bm, physics, dsfont, nicefrac, braket}
\usepackage{graphicx, float, appendix, booktabs, microtype, footnote}
\usepackage[colorlinks,citecolor=myblue,linkcolor=myblue,urlcolor=myblue]{hyperref}
\usepackage{listings}
\usepackage[normalem]{ulem}
\usepackage[table]{xcolor}
    \definecolor{myblue}{RGB}{3,70,143}
    \definecolor{mypink}{RGB}{255,151,151}
    \definecolor{mypurple}{RGB}{164,0,102}
    \definecolor{myteal}{RGB}{0,157,158}
\usepackage[pass]{geometry}

\newcommand{\meas}{\mathcal{M}}
\newcommand{\eps}{\varepsilon}
\newcommand{\pairset}{\mathcal{S}}
\newcommand{\eye}{\mathds{1}}

\DeclareMathOperator*{\argmax}{argmax}

\newtheorem{theorem}{Theorem}
\newtheorem{lemma}[theorem]{Lemma}

\newcommand{\supp}{\cite{*[{See Supplemental Material at }] [{ for the appendices, which includes Refs.~[42--49].}] supp}}

\begin{document}

\title{Optimal Overlapping Tomography}

\author{Kiara Hansenne}
\email{kiara.hansenne@ipht.fr}
\affiliation{Naturwissenschaftlich-Technische Fakult\"at, Universit\"at Siegen, Walter-Flex-Stra{\ss}e 3, 57068 Siegen, Germany}
\affiliation{Institut de Physique Théorique, Université Paris-Saclay, CEA, CNRS, 91191 Gif-sur-Yvette, France}

\author{Rui Qu}
\email{rui.qu@ntu.edu.sg}
\affiliation{Division of Physics and Applied Physics, School of Physical and Mathematical Sciences, Nanyang Technological University, Singapore 637371, Singapore}
\affiliation{Centre for Quantum Technologies, National University of Singapore, Singapore 117543, Singapore}

\author{Lisa T. Weinbrenner}
\affiliation{Naturwissenschaftlich-Technische Fakult\"at, Universit\"at Siegen, Walter-Flex-Stra{\ss}e 3, 57068 Siegen, Germany}

\author{Carlos de Gois}
\affiliation{Naturwissenschaftlich-Technische Fakult\"at, Universit\"at Siegen, Walter-Flex-Stra{\ss}e 3, 57068 Siegen, Germany}

\author{Haifei Wang}
\affiliation{Division of Physics and Applied Physics, School of Physical and Mathematical Sciences, Nanyang Technological University, Singapore 637371, Singapore}
\affiliation{Centre for Quantum Technologies, National University of Singapore, Singapore 117543, Singapore}

\author{Yang Ming}
\affiliation{Division of Physics and Applied Physics, School of Physical and Mathematical Sciences, Nanyang Technological University, Singapore 637371, Singapore}

\author{Zhengning Yang}
\affiliation{Division of Physics and Applied Physics, School of Physical and Mathematical Sciences, Nanyang Technological University, Singapore 637371, Singapore}

\author{Pawe{\l} Horodecki}
\affiliation{Faculty of Applied Physics and Mathematics, Gda\'nsk University of Technology, 80-233 Gda\'nsk, Poland}
\affiliation{International Centre for Theory of Quantum Technologies, University of Gda\'nsk, Jana Ba\.zy\'nskiego 1A, 80-309 Gda\'nsk, Poland}

\author{Weibo Gao}
\email{wbgao@ntu.edu.sg}
\affiliation{Division of Physics and Applied Physics, School of Physical and Mathematical Sciences, Nanyang Technological University, Singapore 637371, Singapore}
\affiliation{Centre for Quantum Technologies, National University of Singapore, Singapore 117543, Singapore}

\author{Otfried G\"uhne}
\email{otfried.guehne@uni-siegen.de}
\affiliation{Naturwissenschaftlich-Technische Fakult\"at, Universit\"at Siegen, Walter-Flex-Stra{\ss}e 3, 57068 Siegen, Germany}

\begin{abstract}
    Characterising large-scale quantum systems is central to fundamental physics and essential for applications of quantum technologies. 
    While a full characterisation requires exponentially increasing resources, focusing on application-relevant information can often lead to significantly simplified analysis. 
    Overlapping tomography is such a scheme, allowing {one} to obtain all the information contained in specific subsystems of multiparticle quantum systems in an efficient manner, but the ultimate limits of this approach remain elusive. 
    We present protocols for overlapping tomography that are optimal with respect to the number of measurement settings. 
    First, by providing algorithmic approaches based on graph theory we find the minimal number of Pauli settings, relating overlapping tomography to the problem of covering arrays 
    in combinatorics.
    This significantly reduces the 
    number of measurement settings,
    showing for instance that two-body 
    overlapping tomography of 
    nearest neighbours in qubit systems with planar topologies
    can always be performed with nine Pauli settings. 
    Second, we prove that 
    using general projective measurements, 
    all $k$-body marginals can be reconstructed with only $3^k$ settings, independently of the system size. 
    Finally, we demonstrate the practical applicability of our methods in a six-photon experiment. 
    Our results will find applications in 
    learning noise and interaction patterns in quantum computers as 
    well as characterising fermionic systems in quantum chemistry. 
\end{abstract}

\maketitle

    Extracting relevant information from large-scale quantum systems is key to advancing quantum technologies. 
    Depending on the scenario, this may involve certifying quantum computers \cite{broadbent2018howtoverify,Gheorghiu2018verification,carrasco2023gaining,proctor2024benchmarking}, identifying quantum phase diagrams \cite{Rem2019identifying,kottmann2020unsupervised}, or certifying the topology and robustness of quantum networks \cite{weinbrenner2024certifying,mao2024certifying}. 
    Currently, efforts are devoted to Hamiltonian learning \cite{knauer2017experimental, evans2019scalable, gebhart2023learning, huang2023learning, olsacher2024hamiltonian, stilck2024efficient}, quantum noise characterisation \cite{samach2022lindblad, vandenberg2023probabilistic, jaloveckas2023efficient, berg2023techniques, wagner2023learning, raza2024online}, or quantum state analysis \cite{huang2020predicting,nguyen2022optimizing}, all of which share a feature of locality.
    Physical Hamiltonians and quantum noise models are often local, involving only few-particle interactions, while quantum states are often well characterised by their few-body marginals. 
    This holds for instance for ground states of gapped few-body Hamiltonians \cite{haselgrove2003quantum, huber2016characterizing, karuvade2019uniquely}, states from shallow circuits \cite{yu2023learning}, injective projected entangled pair states \cite{perez2008peps, molnar2018generalization, cirac2021matrix}, and in some cases, generic pure states \cite{linden2002almost, jones2005parts}. 
    Developing methods to extract such local information is vital, particularly given recent experiments where switching measurement settings requires additional effort \cite{friis2018observation}. 

    The concept of overlapping tomography addresses this problem.
    In its original formulation \cite{cotler2020quantum, bonetmonroig2020nearly, garciaperez2020pairwise}, the authors aimed to reconstruct every $k$-body marginal of an $n$-qubit state $\varrho$ using only Pauli settings, $\set{X, Y, Z}^{\otimes n}$, and showed that this requires a measurement budget increasing only logarithmically in $n$. 
    This was later extended to nearest-neighbour marginals \cite{araujo2022local}, shown to achieve quadratic improvement when considering repetitions of Pauli settings \cite{veltheim2024multiset}, applied to quantum noise characterisation \cite{maciejewski2021modeling, tuziemski2023efficient} and demonstrated experimentally \cite{yang2023experimental}.

    We generalise overlapping tomography in several aspects and present optimal schemes for it.
    First, under the Pauli measurement restriction, we propose a graph-theoretic technique to construct minimal measurement schemes that recover all $k$-body marginals.
    This significantly improves upon current solutions and, importantly, extends to cases where only specific marginals are requested. 
    For example, when marginals of nearest neighbours in planar lattices are needed, we show that nine Pauli settings are sufficient, irrespective of the system size. 
    Second, relaxing the Pauli measurement restriction, we consider local qubit rotations and computational basis measurements, which is routine for all major experimental platforms.
    {With projective measurements, a minimum of $3^k$ settings is required for the state tomography of $k$-qubit systems.
    Here, we show that $k$-body overlapping tomography of $n$-qubit systems can be performed with that minimum of $3^k$ settings, independent of system size $n$.}
    To demonstrate feasibility and scalability, we fully characterise all two-body marginals of a six-photon entangled Dicke state using $12$ optimal Pauli settings and nine optimal non-Pauli settings.
    This nearly halves the requirements of a recent demonstration \cite{yang2023experimental} while generalising to non-Pauli measurements.

{\it Optimality in the Pauli scheme---}For Pauli measurements, overlapping tomography reduces to identifying a set of $n$-body Pauli settings that allows for 
    the reconstruction of all $k$-qubit marginal states. 
    While a naive count suggests $3^k \binom{n}{k}$ measurement settings, this approach is highly redundant, as many settings are effectively identical across $k$-qubit subsets.
    Overlapping tomography exploits these overlaps to construct much smaller measurement sets.
    To give a simple example, the nine Pauli settings
    $XXXX$,
    $ZYYX$,
    $YZZX$,
    $YYXY$,
    $XZYY$,
    $ZXZY$,
    $ZZXZ$,
    $YXYZ$, and
    $XYZZ$ (omitting tensor product symbols)
    suffice to reconstruct all two-qubit marginals, as for every pair of qubits, the nine combinations of Pauli operators occur. 
    The cardinality of such minimal Pauli sets is denoted by $\phi_k(n)$. 
    In the example, the set is optimal, so $\phi_2(4)=9$.

    Our main observation is that this problem can be mapped to a graph-theoretic formulation solvable via binary programming.
    For simplicity, we focus on reconstructing all two-body marginals of an $n$-qubit state, with generalisations to subsets of $k$-body marginals detailed in Appendix \ref{app-sec:arbitrary_connect} of Supplemental Material (SM) \supp \nocite{appel1977solution, bretto2013hypergraph, pk2019distance, colbourncovering, torresuniform, markowits1952portfolio, hedayat1999orthogonal, sloanetables}. 

    First, we construct a graph where the edges represent the two-qubit Pauli operators required for reconstructing the marginals. 
    Each qubit is represented by three vertices (one for each Pauli operator), and vertices corresponding to different qubits are fully connected, forming the complete $n$-partite graph with three nodes per party, $K_{n,3}$ (see Fig.~\ref{fig:3qubit} for a three-qubit example).
    \begin{figure}
        \centering
        \includegraphics[width=0.95\columnwidth]{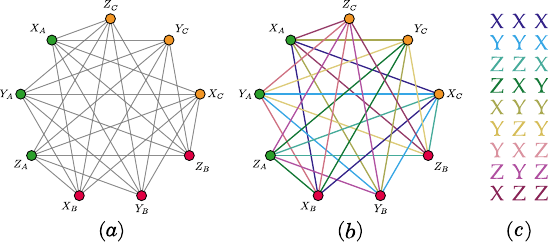}
        \caption{
        (a) Measurement graph $K_{3,3}$ for three qubits. 
        Each qubit is represented by three vertices corresponding to the Pauli operators $X$, $Y$, and $Z$. 
        Edges represent two-body Pauli operators, and cliques (triangles) correspond to three-qubit Pauli settings.
        (b) A minimal edge clique covering of $K_{3,3}$. 
        Each colour groups edges into a triangle, minimising the total number of triangles needed to cover all edges.
        (c) Three-qubit Pauli settings derived from the clique covering, ensuring all two-body Pauli operators for each qubit pair are included. Colours match the triangles in (b).
        }
        \label{fig:3qubit}    
    \end{figure}
    Here, $n$-body Pauli settings correspond to fully connected vertex sets (cliques) in $K_{n,3}$. 
    Thus, finding minimal Pauli settings is equivalent to finding the smallest set of cliques covering all graph edges.

    This can be expressed as a binary optimisation program. 
    First, let $C$ and $E$ denote the sets of cliques and edges of the graph $K_{n,3}$, respectively. 
    We then associate each clique $c \in C$ to a variable $z_c \in \set{0, 1}$. Whenever a 
    clique is active ($z_c = 1$), then the measurement associated to $c$ is part 
    of the solution. 
    Each clique contains several edges, and if clique $c$ is active, 
    all its edges are covered. 
    For a solution to be valid, all the edges of $K_{n,3}$ must be covered. 
    To guarantee this, we associate each edge $e \in E$ to a vector $(e_1, \ldots, 
    e_{\abs{C}})$, where $e_c = 1$ if the edge $e$ is in the clique $c$, otherwise 
    $e_c = 0$. 
    Then, the condition can be expressed by requiring that $\sum_{c \in C} z_c e_c 
    \geq 1$ for any $e \in E$, meaning that for each edge of the graph, at least one 
    clique that contains it must be in the solution. 
    Altogether, we obtain
    \begin{subequations}\label{eq:mip}
        \begin{align}
            \phi_2(n) = &
            \min_{\set{ z_c }_{c \in C}} \;\;\; \sum_{c \in C} z_c \\
            &\text{subject to:} \;\;\; z_c \in \set{0, 1}, \,\forall c \in C, \\
            & \phantom{\text{subject to:}} \;\;\; \sum_{c \in C} z_c e_c \geq 1, \,\forall e \in E .
        \end{align}
    \end{subequations}
    Equation \eqref{eq:mip} can be modelled using modelling languages such as JuMP \cite{Lubin2023} and solved by standard integer programming solvers (see, e.g., Gurobi \cite{gurobi}).
    The minimal Pauli settings are obtained from the resulting values of $z_c$. 
    The solutions improve significantly over previous constructions. For example, the 
    four-qubit solution from above improves the best previous Pauli set, which contains $12$ 
    Pauli settings to perform two-body overlapping 
    tomography \cite{garciaperez2020pairwise}. 
    For six qubits, we obtain $\phi_2(6)=12$, whereas in
    Refs.~\cite{cotler2020quantum, bonetmonroig2020nearly} and in
    Ref.~\cite{garciaperez2020pairwise}, $21$ and $15$ settings are needed respectively 
    (see Appendix \ref{app-sec:min_pauli_sets} of SM \supp \ for explicit minimal Pauli sets).
    
    We add that this graphical formulation is a specific instance of the edge clique 
    covering problem, which has long been studied in graph theory. For instance, de 
    Caen \textit{et al.\@} showed in 1985 that $\phi_2(n) \leq 6 \lceil \log_3(n) \rceil + 3$ \cite{decaen1985clique}, which exactly matches the size of the construction presented in 
    Ref.~\cite{garciaperez2020pairwise}. 
    The general formulation of the edge clique covering problem, however, is known to be NP-complete \cite{orlin1977contentment}.
    {The similar problem of vertex clique covering has also been explored in the context of Pauli measurement scheduling; see, for instance, Refs.~\cite{verteletskyi2020measurement, veltheim2024multiset}.} 

    The combinatorial problem underlying overlapping tomography with Pauli measurements corresponds to covering arrays in combinatorial design theory \cite{colbourn2004combinatorial, torres2013survey}.
    Several facts about the minimal Pauli sets can thus be immediately borrowed, such as bounds on $\phi_k(n)$ and explicit constructions. 
    As examples, it is known that $\phi_k(k+1) = 3^k$, and that $\phi_k(n) \leq (k-1)\Big[\log(\frac{3^k}{3^k-1})\Big]^{-1} \log (n) [1+o(1)]$ \cite{colbourn2004combinatorial, torres2013survey}. 
    However, minimal covering arrays are notoriously hard to construct: when $k=2$, covering arrays with three symbols (here, $X$, $Y$, and $Z$, the three Pauli operators) are only known up to $n=20$ parties \cite{kokkala2020structure}. 
    The resulting $\phi_2(n)$ are presented in Fig.~\ref{fig:nbr_of_settings}, and compared to previous constructions \cite{cotler2020quantum, bonetmonroig2020nearly, garciaperez2020pairwise}.
    \begin{figure}
        \centering
        \includegraphics[width=0.9\columnwidth]{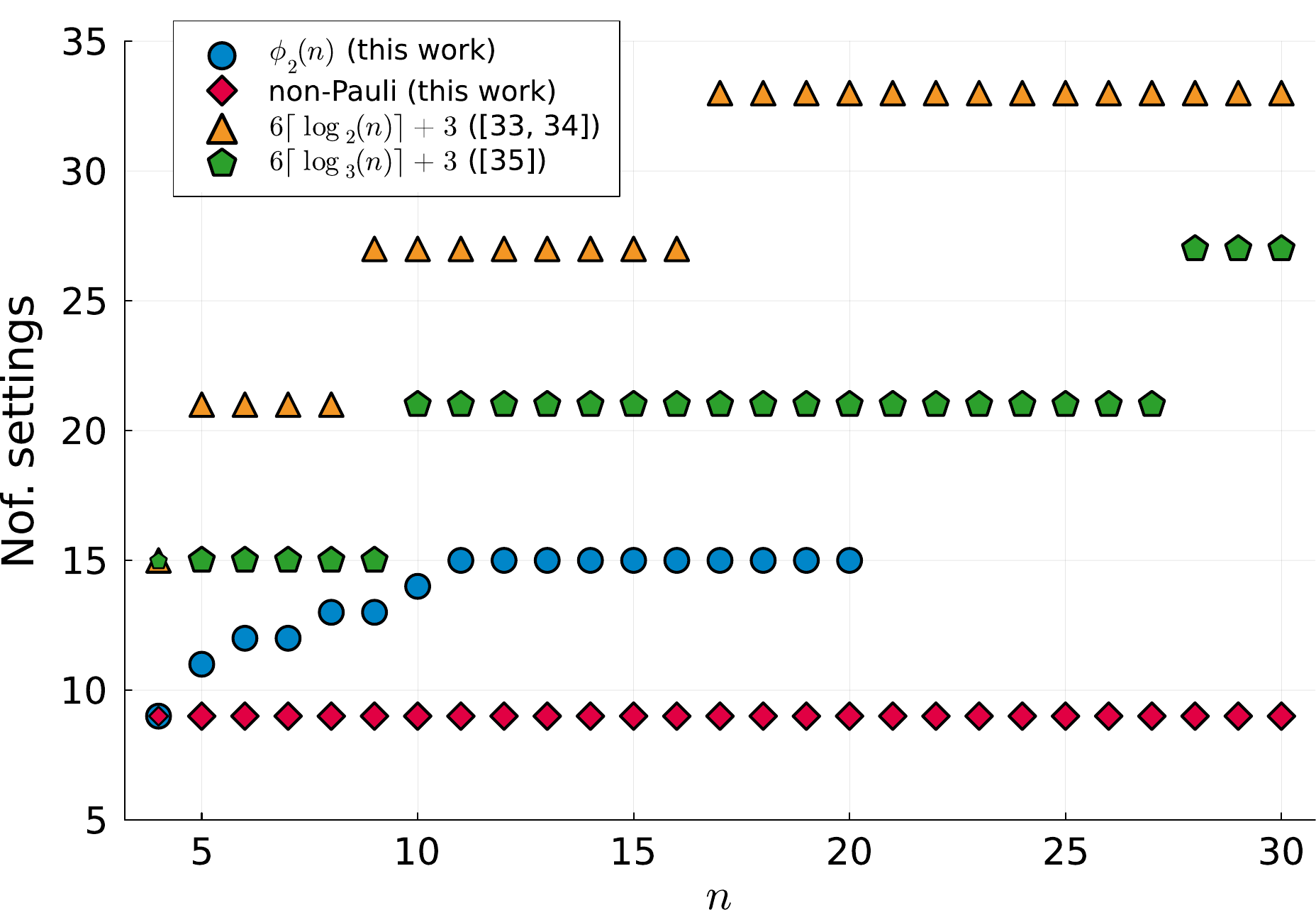}
        \caption{
        Number of measurement settings needed for two-body overlapping tomography of $n$ qubits ($k=2$).
        Blue circles: minimal number of Pauli settings possible, corresponding to covering array numbers \cite{kokkala2020structure}.
        Red squares: minimal number of projective settings, $3^2 = 9$.
        Orange triangles: number of Pauli settings needed in the construction of Refs.~\cite{cotler2020quantum, bonetmonroig2020nearly}.
        Green pentagons: number of Pauli settings needed in the construction of Ref.~\cite{garciaperez2020pairwise}.}
        \label{fig:nbr_of_settings}    
    \end{figure}
    The connection to covering arrays for the special case $k=2$ has been observed before \cite{hou2006controlling, maltais2009covering, ronneseth2009merging, danziger2009covering}, and a hypergraph generalisation to $k>2$ is presented in Appendix \ref{app-sec:k-body_tomo} of SM \supp. 
    Note that Ref.~\cite{berg2023techniques} pointed to a connection between covering arrays and graph colouring to the learning of noise models, which turns out to be equivalent to the problem of partial tomography for all $k$-body marginals.

    The advantage of our graphical formulation is twofold. First, the binary program 
    presented in Eq.~\eqref{eq:mip} can directly be extended to higher values of $k$ and 
    can be used to compute, for instance, $\phi_3(4) = 27$ and $\phi_3(5) = \phi_3(6) = 33$ 
    (see Appendix \ref{app-sec:min_pauli_sets} of SM \supp). 
    To the best 
    of our knowledge, it had not been previously shown that covering arrays corresponding 
    to $\phi_3(5) = 33$ were optimal. 
    Second, the graphical formulation is highly flexible: 
    with appropriate modifications, it can be used to obtain minimal Pauli sets to reconstruct 
    not all marginals, but rather only a subset of them. This problem has been coined local 
    overlapping tomography \cite{araujo2022local}. In the graphical formulation, we 
    define the $n$-vertex connectivity graph $G$ where vertices are connected if the marginal 
    of the corresponding qubits is desired.  
    We explore this in detail in Appendix \ref{app-sec:arbitrary_connect} of SM \supp \ and, notably, we prove that for systems having a connectivity 
    graph with a chromatic number of at most four (such as for instance all planar graphs), 
    two-body overlapping tomography can always be performed with nine Pauli settings, 
    regardless of the number of qubits.

{\it Optimality beyond Pauli measurements---}If the local settings are restricted to Pauli measurements, the number of measurement settings increases with $n$. 
    To improve this, more general settings than just Pauli measurements can be allowed locally. 
    Indeed, this approach ensures that the minimal number of $n$-qubit measurement settings, $3^k$, can always be reached.

    Reconstructing every two-body marginal state can be done by obtaining the 
    expectation values of the nine Pauli observables $\set{X,Y,Z}^{\otimes 2}$ 
    for all pairs of qubits $\set{i,j}$. In terms of Bloch vectors, this means 
    that the two parties choose their measurement directions on the Bloch 
    sphere to be exactly the standard basis $\set{\vec{e}_m}_{m=1}^3$ in
    three-dimensional real space, and measure all possible combinations 
    \begin{equation}
    \meas_{m,l}=  \vec{e}_m \otimes \vec{e}_l,
    \end{equation}
    where the observables can be obtained from this shorthand notation by
    setting $\vec{v} \mapsto v_1 X + v_2 Y + v_3 Z$.

    To obtain tomographically complete information on the two-qubit marginal, it is not important that the nine vectors $\vec{e}_m \otimes \vec{e}_l$ are orthogonal.
    They simply need to be linearly independent and thus form a basis.
    So, the key idea is to randomly choose nine observables for each party instead of 
    the fixed Pauli settings. 
    We denote the respective local measurement directions on the 
    Bloch sphere by $\vec{v}_\alpha^{(i)}$ for all $\alpha \in \set{1, \dots , 9}$, 
    where $i$ labels the party. 
    Intuitively, if those vectors are chosen independently 
    and according to the uniform distribution on the unit sphere, the nine product vectors 
    $\vec{v}_\alpha^{(i)} \otimes \vec{v}_\alpha^{(j)}$ for the pair $\set{i, j}$ should 
    be linearly independent with unit probability.
    Since each party chooses their measurement directions randomly, the above argument holds 
    for every possible pair of qubits.

    This reasoning suggests that choosing nine random measurement directions $\vec{v}_\alpha^{(i)}$ 
    for each of the $n$ qubits and measuring
    \begin{equation}
        \meas_\alpha =  \bigotimes_{i=1}^n \vec{v}_\alpha^{(i)}, \quad \alpha = 1,\dots, 9 
    \end{equation}
    is sufficient to reconstruct all the two-body marginal states of the $n$-qubit system. 
    In Appendix \ref{app-sec:min_settings} of SM \supp \ we show that this intuition indeed holds true for 
    arbitrary $k$. 
    Remarkably, the proof holds for arbitrary local dimension $d$, which
    shows that $k$-body overlapping tomography of qudit systems 
    can be performed with $(d^2-1)^k$ settings. 
    Notice also that the resulting measurements are projective and local, and therefore do 
    not impose extra experimental requirements.
    For completeness, this minimal number of settings is also plotted in Fig.~\ref{fig:nbr_of_settings}.

    Although random measurement directions yield tomographically complete settings for all marginals, some settings perform better than others regarding statistical errors.
    Measurement directions that are too close to each other can lead to large variances in the reconstructed states, particularly when the state is nearly orthogonal to all the measurement directions.
    To assess the quality of a particular set of measurement directions, we use confidence regions informed by measurements variances. 
    Specifically, we use the confidence region $C_A$ from Ref.~\cite{degois2023userfriendly} as a figure of merit, which is straightforward to construct for any measurement scheme.
    Numerically optimised measurement directions and performance comparisons are provided in Appendices \ref{app-sec:conf_regs_and_num_opt} and \ref{app-sec:nbr_samples} of SM \supp.

{\it Experimental demonstration---}We experimentally perform two-body overlapping tomography on 
    a six-photon Dicke state with three excitations, that is,
    \begin{figure}[t]
        \centering
        \includegraphics[width=\linewidth]{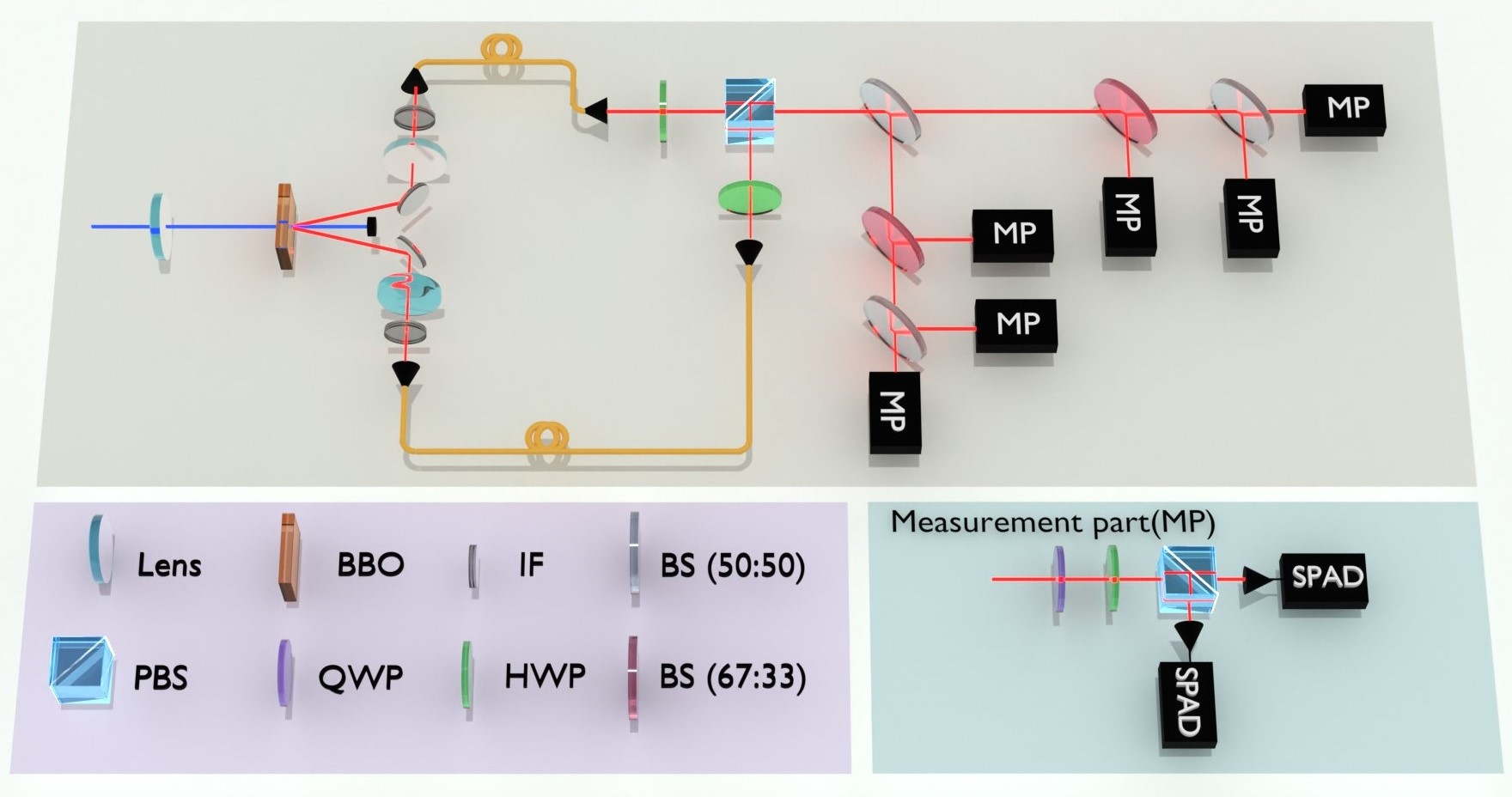}
        \caption{
        Experimental setup for generating and detecting six-photon Dicke state with three excitations.
        Ultraviolet pulses (390 nm, 80 MHz repetition rate, 300 mW average power) are focused onto a type-II beamlike $\beta$-barium borate (BBO) crystal using a lens ($f_1=75$mm) to generate three photon pairs photons simultaneously in the third-order spontaneous parametric down-conversion (SPDC) emissions.
        The photons are recollimated with two lenses ($f_2=100$ mm), spectrally and spatially filtered by interference filters (IFs) 
        ($\Delta\lambda=10$nm) and single-mode fibres, and merged into a single path by a polarising beam splitter (PBS).
        Hong-Ou-Mandel-type interference visibility is $0.92$.
        The six indistinguishable photons are distributed into six output modes using three 50:50 and two 67:33 beam splitters (BSs), achieving a maximal success probability of $5/324$ for detecting one photon in each measurement part (MP).
        Arbitrary polarisation analysis in each output mode is conducted by the MP, composed of a quarter-wave plate (QWP), a half-wave plate (HWP), a PBS, and two single-photon avalanche detectors (SPADs).   
        }       
        \label{fig:setup}    
    \end{figure}
    %
    \begin{equation}
        \ket{D_{6}^{(3)}} = \frac{1}{\sqrt{20}} \sum_{i} \mathcal{P}_{i}(|H H H V V V\rangle),
    \end{equation}
    where $|H\rangle$ ($|V\rangle$) denotes the horizontal (vertical) polarisation and 
    $\sum_{i} \mathcal{P}_{i}(...)$ denotes the sum over all 20 permutations 
    leading to different terms. 
    Dicke states are highly entangled multipartite states, widely recognized for their applications in quantum metrology \cite{toth2012multipartite, saleem2024achieving} and quantum networks \cite{Dickestate2009-zeilinger, chiuri2012experimental, roga2023efficient}.
    Moreover, all their marginal states are entangled \cite{szalay2025dicke}, which motivates using a Dicke state to demonstrate the practicality of overlapping tomography. 

    The setup of the experiment, described in detail in Fig.~\ref{fig:setup}, is registering $\ket{D_{6}^{(3)}}$ 
    with a sixfold coincidence count rate of 7.0 counts per minute. To verify the proper working of our source, we 
    measure the structures in the $X$, $Y$, and $Z$ bases and  
    observe the characteristic features of the Dicke state 
    $\ket{D_{6}^{(3)}}$~\cite{Dickestate2009-weinfurter,Dickestate2014}, which is detailed in Appendix \ref{app-sec:charact_dicke} of SM \supp.

    To conduct the six-qubit full state tomography, $729$ Pauli settings 
    would need to be measured, making the measurement time prohibitively 
    long. 
    Carrying out optimal two-body overlapping 
    tomography in the Pauli scheme requires $\phi_2(6)=12$ settings only (see Appendix \ref{app-sec:min_pauli_sets} of SM \supp \ for
    the specific measurements). 
    We perform these $12$ Pauli measurement settings with an acquisition 
    time for each measurement setting of two hours.  
    By using maximum likelihood estimation (MLE) \cite{MLE2001,altepeter2005photonic}, 
    we reconstruct the physical experimental density matrices of all $15$ 
    two-qubit subsystems $\varrho^{\rm{P}}_{i_1,i_2}$,  where $i_1,i_2\in\set{0,1,...,5}$. Figure \ref{fig:fidelity} (a) depicts the reconstructed marginals of $\varrho^{\rm{P}}_{2,3}$ for the Pauli scheme.
    \begin{figure*}
        \centering
        \includegraphics[width=0.9\linewidth]{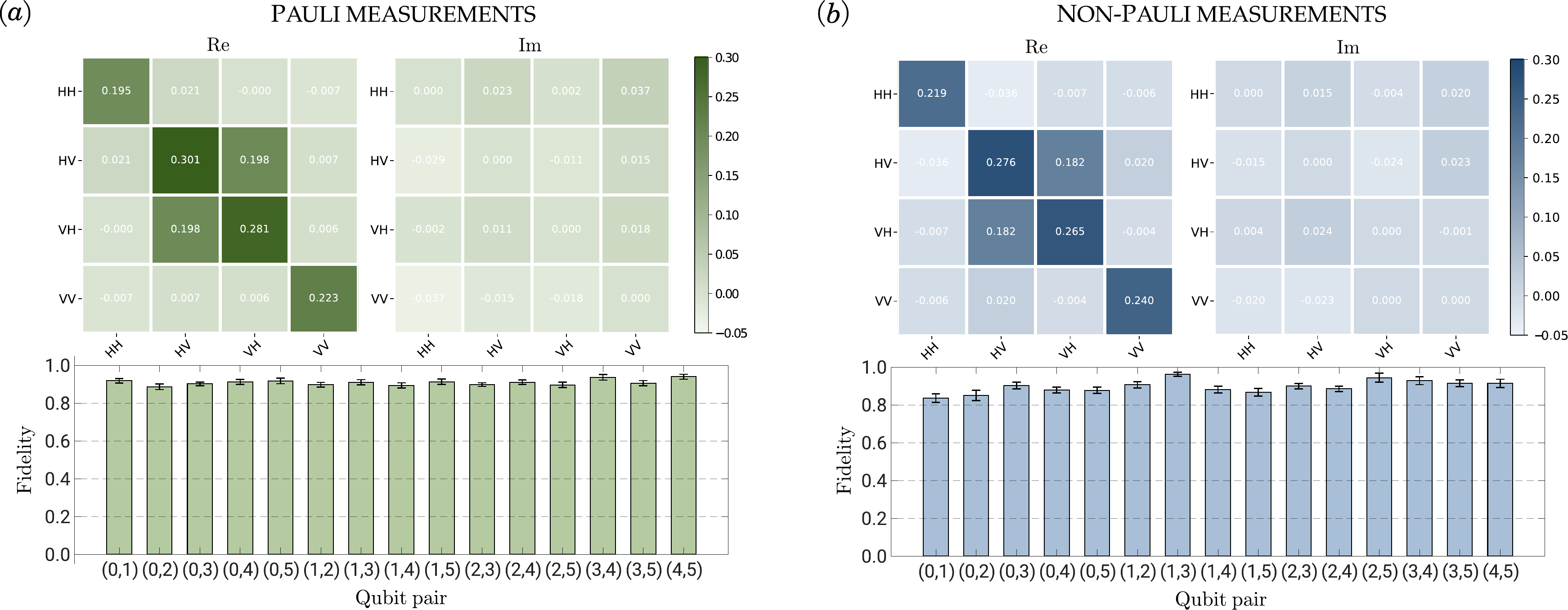}
        \caption{
        Two-qubit reduced states from overlapping tomography of a six-qubit Dicke state.
        (a) Results for the optimal Pauli scheme (12 settings): the real and imaginary parts of the marginal $\varrho_{2,3}^P$ are shown (see Appendix \ref{app-sec:exp} of SM \supp \ for other marginals), along with the two-qubit mixed state fidelities between experimental states $\varrho_{i_1, i_2}^P$ and ideal marginals $\varrho_{i_1, i_2}$ for each pair.
        (b) Similar analysis for the minimal non-Pauli scheme (nine settings).
        In both cases, fidelity error bars indicate $\pm1 \sigma$ uncertainties estimated via Monte Carlo sampling with Poissonian photon statistics.
        }
        \label{fig:fidelity}    
    \end{figure*}
    
    Second, we perform the measurements with the nine non-Pauli directions described in the previous section. The acquisition time for each measurement setting is also set to two hours. Similarly, we reconstructed all two-body marginals $\varrho^{\rm{N}}_{i_1,i_2}$ in the non-Pauli scheme, e.g., $\varrho^{\rm{N}}_{2,3}$ presented in Fig.~\ref{fig:fidelity} (b). 
    The method of reconstruction and all experimental marginals for both schemes are presented in Appendix \ref{app-sec:expe_reconstruct} of SM \supp. 
        
    Subsequently, in order to obtain error estimates, we perform $100$ 
    Poisson distribution samples on the experimental sixfold data using 
    a Monte Carlo approach.
    We obtain the experimental mixed state fidelities   $F_{i_1,i_2}=\left[\operatorname{Tr}\left(\sqrt{\sqrt{\varrho_{i_1,i_2}}\varrho^{\rm{exp}}_{i_1,i_2}\sqrt{\varrho_{i_1,i_2}}}\right)\right]^{2}$ and their error bars, where $\varrho_{i_1,i_2}=\operatorname{Tr}_{j_1,j_2,j_3,j_4}(\ket{D_{6}^{(3)}}\bra{D_{6}^{(3)}})$ denotes the  ideal two-body marginals {and $\set{j_1,j_2,j_3,j_4}=\set{0,1,...,5} \setminus \set{i_1,i_2}$}. As shown in Fig.~\ref{fig:fidelity} (a), for the 
    Pauli scheme, the average experimental fidelity of the $15$ 
    two-body marginals is $0.909$ and the minimum fidelity is $0.887$. 
    The average error bar is $0.013$. By contrast, Ref.~\cite{yang2023experimental} performed 21 Pauli settings on a six-qubit GHZ state for around 80 hours to reconstruct its two-body marginals, achieving an average fidelity of $0.848\pm0.013$.

    Figure \ref{fig:fidelity} (b) shows that 
    the average fidelity of reconstructed two-body marginals for non-Pauli Scheme is $0.896$ 
    and the average error bar is $0.018$. 
    We also observe that requiring a minimal number of nine non-Pauli 
    settings comes at the cost of slightly larger error bars {with the same acquisition time for each setting, which is further discussed in Appendix \ref{app-sec:error_bars} of the SM \supp.} 
    Our findings demonstrate in theory and in practice that 
    overlapping tomography can be used to drastically reduce 
    experimental requirements, both in the number of measurement 
    settings and measurement time.

{\it Discussion---}In this Letter, we have addressed the problem of optimising measurement scheduling for marginal tomography, combining theoretical analysis with experimental validation.
On a purely mathematical side, the connections to graph theory and combinatorics may be used  to tackle problems in quantum information theory from a novel perspective. 
For example, the concept of covering arrays generalises orthogonal arrays, which have recently led to substantial progress in the study of absolutely maximally entangled states \cite{goyeneche2018entanglement}, leading to the 
solution of one of the five central open problems in quantum 
information theory \cite{horodecki2020five, rather2022thirty}.
These connections demonstrate the potential of interdisciplinary approaches in quantum information theory.

On a more practical side, the presented optimal protocols are expected to 
find applications in quantum simulations, such as the characterisation of local 
Hamiltonians \cite{knauer2017experimental, evans2019scalable, gebhart2023learning, huang2023learning, olsacher2024hamiltonian, stilck2024efficient}.
They can also be used for analysing noise and imperfections in realistic quantum computing devices \cite{samach2022lindblad, vandenberg2023probabilistic, jaloveckas2023efficient, berg2023techniques, wagner2023learning, raza2024online}. 
In addition, these techniques can be extended to measure 
few-body correlations in fermionic or bosonic systems, 
as well as to characterise and certify quantum gates in different registers of a quantum computing architecture in 
parallel. 
The general applicability of our approach to systems of any size and geometry highlights partial quantum tomography's potential as a scalable and efficient tool for addressing key challenges in large-scale quantum information processing.

\textit{Acknowledgments---}We thank Adrian Aasen, Maximilian Hess, Jose Este Jaloveckas, Matthias Kleinmann, Chau Nguyen, Lilly Palackal, Martin Pl\'avala, Jonathan Steinberg, Konrad Szyma\'nski, Ewout van den Berg, Pawel Wocjan, and Xiao-Dong Yu for discussions. 
We also acknowledge the usage of ChatGPT, which mentioned the keyword ``covering array'' in relation to one of our combinatorial problems.
This work has been supported by the Deutsche Forschungsgemeinschaft (DFG, German Research Foundation, projects No.\ 447948357 and No.\ 440958198, the Sino-German Center for Research Promotion (Project M-0294), and the German Ministry of Education and Research (Project QuKuK, BMBF Grant No.\ 16KIS1618K). 
R. Q., H. W., Y. M., Z. Y. and W. G. acknowledge Singapore Quantum engineering program (NRF2022-QEP2-02-P14).   
K. H., L. T. W., and C. de G. acknowledge support by the House of Young Talents of the University of Siegen.

\bibliography{bibliography}

\begin{appendix}

\onecolumngrid
\newpage
\newgeometry{margin=9.0em}

\thispagestyle{empty}
\begin{center}
    \textbf{\large Supplemental Material
    for ``Optimal overlapping tomography''}

    Kiara Hansenne,\footnote{kiara.hansenne@ipht.fr} Rui Qu,\footnote{rui.qu@ntu.edu.sg} Lisa T.\@ Weinbrenner, Carlos de Gois, Haifei Wang, Yang Ming, Zhengning Yang, Pawe\l \ Horodecki, Weibo Gao,\footnote{wbgao@ntu.edu.sg} and Otfried Gühne \footnote{otfried.guehne@uni-siegen.de} 
\end{center}

\section{Pauli tomography for arbitrary sets of marginals} \label{app-sec:arbitrary_connect}
    In this appendix, we thoroughly discuss how to obtain tomographically complete Pauli sets for given subsets of qubits.
    In the first section, we focus on two-body partial tomography ($k=2$) and extend the graph construction and binary program of the main text to cases where not all marginals of a given size are needed.
    We then generalise these results to higher $k$s, and illustrate this by computing a minimal Pauli set for a seven-qubit system where certain three-qubit marginals need to be reconstructed. 

    For large systems, minimal Pauli sets might be impractical to compute due to the excessive number of variables in the binary program.
    Therefore, we present two alternative approaches for finding Pauli sets and discuss their optimality.
    The first approach relies on graph colouring and is well suited for relatively small sets of marginals, while the second approach involves an explicit construction for the reconstruction of all marginals for which we provide a comparison with previously existing methods.

    Finally, we provide the minimal Pauli sets that have been computed and used in this manuscript.
    
    \subsection{Two-body marginal tomography}
        We start by addressing the problem of two-body partial tomography ($k=2$) where we are interested in some (i.e.\@, not necessarily all) of the two-qubit marginals of an arbitrary $n$-qubit state $\varrho$.
        
        In partial tomography, the specific marginals of $\varrho$ that must be obtained are determined by the physical problem under consideration. 
        In some instances, we may require all of the $\nicefrac{n(n-1)}{2}$ two-qubit marginals, while in other cases, only a subset may be relevant, such as the nearest neighbours in many-body systems.
        We encode this information in a connectivity graph $G$ with $n$ vertices. Each of these vertices is associated with one of the $n$ qubits,
        and we connect the vertices $i$ and $j$ if and only if the marginal $\varrho_{i,j} = \text{tr}_{[n] \setminus i,j}(\varrho)$ is wanted, with $[n] = \{1,\dots, n\}$. 
    
        For the sake of concreteness, let us first focus on the case of three qubits with a line connectivity graph (Fig.~\ref{app-fig:3qubit_graph} (a)), 
        which means we must reconstruct the marginal states $\varrho_{1,2} = \text{tr}_3(\varrho)$ and $\varrho_{2,3} = \text{tr}_1(\varrho)$. 
        We are thus looking for a minimal Pauli set that covers the two-body Pauli operators for the pair $\{1,2\}$ and for the pair $\{2,3\}$.
    
        \begin{figure*}[h]
            \centering
            \includegraphics[width=.9\linewidth]{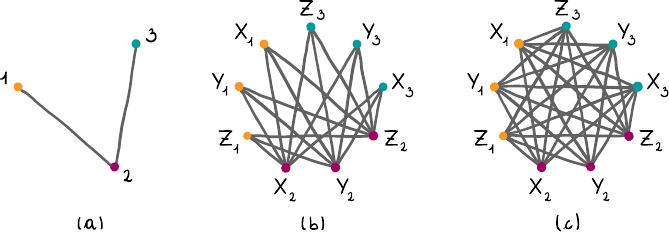}     
            \caption{%
            (a) Connectivity graph $G$ of three qubits. 
            The set of edges $\{\{1,2\}, \{2,3\}\}$ represents the two-qubit marginals that are desired, i.e., we aim to reconstruct the marginals $\varrho_{12}$ and $\varrho_{23}$ of a three-qubit state $\varrho$. 
            (b) Covering graph $G^{\times 3}$ of $G$. 
            Each edge represents a two-body Pauli operator that is needed to recover the two-qubit marginals.
            For instance, the expectation value of $X_1Y_2$ is required and therefore the edge connecting the vertices $X_1$ and $Y_2$ is drawn in the covering graph.
            On the other hand, the expectation value of $X_1Z_3$ is not needed, and the vertices $X_1$ and $Z_3$ are not connected.
            (c) Measurement graph $K_{3,3}$. 
            Each triangle represent a possible three-qubit Pauli setting.
            As an example, the vertices $Y_1$, $X_2$, and $X_3$ are pairwise connected and represent the measurement setting $Y_1X_2X_3$.
            From the measurement data of $Y_1X_2X_3$, it is possible to obtain the expectation values of $Y_1X_2$, $Y_1X_3$, and $X_2X_3$.
            We therefore say that the three-qubit Pauli setting $Y_1X_2X_3$ covers $Y_1X_2$, $Y_1X_3$, and $X_2X_3$, inspired by the fact that the triangle connecting the vertices $Y_1$, $X_2$, and $X_3$ covers the edges $Y_1X_2$, $Y_1X_3$, and $X_2X_3$.
            }
            \label{app-fig:3qubit_graph}
        \end{figure*}
    
        This can easily be described with the help of two additional graphs.
        First, we construct a graph in which the edges represent the two-qubit settings required for reconstructing the marginals. 
        So, in this case, the eighteen two-body Pauli operators for the pairs of qubits $\{1,2\}$ and $\{2,3\}$.
        To do that, we instantiate three vertices per qubit (each of them representing a Pauli operator), and connect the vertices for which the corresponding two-body Pauli is required (Fig.~\ref{app-fig:3qubit_graph} (b)). 
        It is called the covering graph and we define the notation $G^{\times 3}$.
        
        The last graph is used to represent all possible $n$-qubit Pauli settings. 
        It has the same set of vertices as $G^{\times 3}$, and two vertices are connected if and only if they represent single Pauli operators corresponding to different qubits, resulting in a complete $n$-partite graph $K_{n,3}$.
        Therein, an $n$-qubit Pauli setting is represented by the complete subgraph (clique) on the vertices corresponding to the Pauli setting, similarly to the main text.
        For instance, in our example, the three-qubit Pauli setting $Y_1 X_2 Z_3$ is represented by the clique with vertices $Y_1$, $X_2$ and $Z_3$, thus forming a triangle in the measurement graph (Fig.~\ref{app-fig:3qubit_graph} (c)).  
        Note that each of these cliques covers three two-body Pauli settings, which in this case are $Y_1X_2$, $Y_1Z_3$ and $X_2Z_3$.
        
        From this description, it is obvious that finding a minimal Pauli set that recovers the marginals specified in $G$, is equivalent to finding the minimal set of maximal cliques of $K_{3,3}$ that can cover all the edges of $G^{\times 3}$. 
        Furthermore, this formulation can easily be extended to $n$ parties and arbitrary connectivity graphs $G$: 
        First, we let the edges of $G^{\times 3}$ --- also denoted by $E(G^{\times 3})$ --- represent all the two-body Pauli operators whose expectation values need to be known. 
        Then, the maximal cliques of $K_{n, 3}$ --- henceforth $C(K_{n,3})$ --- represent all the possible $n$-qubit Pauli settings.
    
        In the main text, we aimed at reconstructing all the two-qubit marginal states, and thus needed the expectation values of all two-body Pauli operators. 
        In that case, the covering graph is also the $n$-partite complete graph $K_{n, 3}$ and we recover the main result: 
        The covering graph is the same as the measurement graph, hence we only need one graph to solve the problem.
    
        The binary program of Eq.~\eqref{eq:mip} in the main text can easily be extended by taking the edges and the cliques of the covering graph and the measurement graph respectively. 
        The cardinality $\phi_2$ of a minimal set of cliques of $K_{n, 3}$ needed to cover the edges of $G^{\times 3}$ can be written as a function of $G$, as $G$ contains all the information of the problem.
        We can determine $\phi_2(G)$ with a binary program that is a direct extension of
        Eq.~\eqref{eq:mip},
        \begin{subequations}\label{app-eq:mip}
        \begin{align}
            \phi_2(G) = \;
            &\min_{\{ z_c \}_{c \in C(K_{n,3})}} \sum_{c \in C(K_{n,3})} z_c \\
            &\;\;\;\text{subject to} \nonumber \\
            &\quad\;\;\; z_c \in \{0, 1\}, \,\forall c \in C(K_{n,3}) \\
            &\quad\;\;\; \sum_{c \in C(K_{n,3})} z_c e_c \geq 1, \,\forall e \in E(G^{\times 3}) ,
        \end{align}
        \end{subequations}
        where $e_c = 1$ if $e \in c$, and $0$ otherwise. 
        Similarly to the main text, each variable $z_c$ is an indicator variable indexed by a clique $c$ of $K_{n,3}$. 
        A clique is part of the optimal solution (i.e., active) if and only if $z_c = 1$. 
        The second constraint means that for every edge $e$ of $G^{\times 3}$, there must be at least one clique $c$ which contains $e$ and is active. 
        Thus, it guarantees that all edges of $G^{\times 3}$ are covered by the active cliques.
        
        Since $\abs{C(K_{n,3})} = 3^n$, large instances may become too expensive to compute. 
        However, we show in Sec.~\ref{app-sec:colouring_construction} that for many physically motivated connectivities, the minimal Pauli set can be reduced to a small instance of the program, and we discuss in Sec.~\ref{app-sec:recursion} how minimal Pauli sets for a few parties can be extended to larger cases.

    \subsection{$k$-body partial tomography} \label{app-sec:k-body_tomo}
        Although most previous methods focused on two-body partial tomography \cite{cotler2020quantum, bonetmonroig2020nearly, garciaperez2020pairwise}, some physical problems may require $k$-body marginals of higher order. 
        In this section, we show how the graph formulation for two-body partial tomography presented in the previous section can be generalised to $k$-body marginals by using hypergraphs.
    
        \begin{figure}
            \centering
            \includegraphics[width=.4\linewidth]{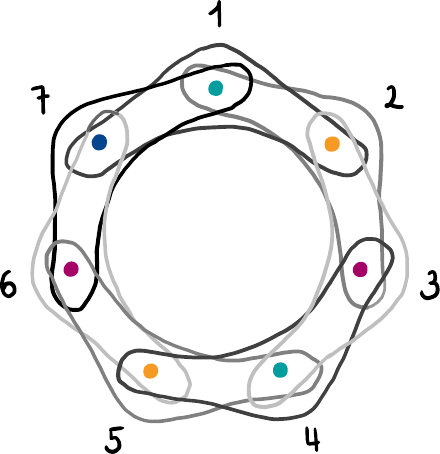}
            \caption{%
            Connectivity hypergraph $H_7$ of seven qubits in a ring where the marginals of each consecutive triplet of qubits are desired, i.e., $\varrho_{123}$, $\varrho_{234}$, $\varrho_{345}$, $\varrho_{456}$, $\varrho_{567}$, $\varrho_{671}$, and $\varrho_{712}$.
            The grey shades of the edges are for better readability.
            }
            \label{app-fig:hypergraph_example}
        \end{figure}
        
        To do that, we start by encoding the desired marginals in a connectivity hypergraph $H$ where each hyperedge connects $k$ vertices representing the qubits from the marginal states.
        As an example, Fig.~\ref{app-fig:hypergraph_example} depicts the hypergraph $H_7$ representing seven qubits in a ring configuration where, for each qubit $i$, we want to have access to the three-qubit marginal of the triplet $\{i-1, i, i+1\}$.  
        The covering hypergraph $H^{\times 3}$ follows the same idea than in the case of $k=2$: Each hyperedge represents a $k$-body Pauli operator whose expectation value is required to reconstruct the marginals. 
        Lastly, the cliques (for hypergraphs, a clique is a subgraph in which any $k$ vertices are connected by a hyperedge) of the measurement graph $K_{n,3,k}$ represent all possible $k$-body Pauli settings.    
        Similarly to the case of $k=2$, an $n$-qubit Pauli setting is a maximal clique in $K_{n,3,k}$.
        Both hypergraphs $H^{\times 3}$ and $K_{n,3,k}$ have $3n$ vertices, one for each single-qubit Pauli setting.
    
        Straightforwardly, optimal $k$-body partial tomography is equivalent to finding a minimal edge covering of $H^{\times 3}$ with cliques of $K_{n, 3, k}$.
        Whereas hypergraphs quickly become cumbersome to draw on paper, Eq.~\eqref{app-eq:mip} extends directly by mapping $K_{n,3} \mapsto K_{n,3,k}$ and $G^{\times 3} \mapsto H^{\times 3}$. 
        Using this extension, we computed $\phi_3(H_7) = 27$ and the corresponding minimal Pauli set is presented in Sec.~\ref{app-sec:min_pauli_sets}. 
        In the next section, we discuss how this result can also be achieved through an explicit construction of the settings.
        
        As with the formulation for $k=2$, increasing the number of parties leads to a considerable computational cost. 
        Nevertheless, minimal Pauli sets for overlapping tomography of all three-body marginals for systems up to six qubits were computed, and are presented in Sec.~\ref{app-sec:min_pauli_sets}.

    \subsection{Reduction for large number of qubits} \label{app-sec:colouring_construction}
        The partial tomography problem can be mapped to a graph covering problem, which can in turn be formulated as a binary program and therefore solved exactly. 
        However, for large connectivity graphs, those techniques can quickly reach their limits regarding what can actually be solved by standard computers. 
        Fortunately, for many physically motivated classes of connectivity graphs, we show in this section how the minimal Pauli sets can be mapped to small instances of the binary program \eqref{app-eq:mip} and optimally computed.
        Again, we start by presenting an example for the sake of clarity, and move on to the general case at the end of the section.
        
        Consider $16$ qubits in a square lattice configuration, in which we aim at reconstructing the two-body marginals of each pair of first and second neighbours (see Fig.~\ref{app-fig:grid_ex} for the connectivity graph $G_{16}$).
        \begin{figure}
            \centering
            \includegraphics[width=.35\linewidth]{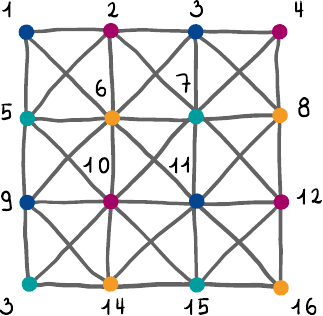}
            \caption{Connectivity graph $G_{16}$ of $16$ qubits in a square grid, where the two-body marginals of the first and second neighbours are required.}
            \label{app-fig:grid_ex}
        \end{figure}
        For 16 qubits, there are $3^{16} > 10^{7}$ possible $16$-qubit Pauli settings, which require as many variables in the binary program and therefore the problem is not solvable with standard computers.
        Regardless, we directly notice that maximal cliques of $G_{16}$ put a lower bound on $\phi_2(G_{16})$. 
        Indeed, for the qubits $1$, $2$, $5$, and $6$, we need to reconstruct the marginals of all six pairs of qubits, therefore at least nine Pauli settings are needed, and we formalise
        \begin{equation} \label{app-eq:lower_bound_g16}
            \phi_{2}(4) \leq \phi_2(G_{16}).
        \end{equation}
        Recall that $\phi_k(n)$ denotes the cardinality of minimal Pauli sets for the tomography of all $k$-body marginals of an $n$-qubit state and $\phi_k(G)$ denotes the cardinality of minimal Pauli sets for the tomography of the $k$-body marginals corresponding to the edges of the (hyper)graph $G$.
        In this particular example, $\phi_{2}(4) = 9$ also happens to be the strict minimum of Pauli settings for two-body tomography.
    
        We proceed by taking a minimal Pauli set for four qubits and associate one colour to each qubit, as shown in Fig.~\ref{fig:grid_meas} (a).
        Each party has nine single-qubit Pauli settings of one colour.
        Then, using the same four colours, we colour the $16$ vertices of the connectivity graph $G_{16}$ in such a way that no adjacent vertices (qubits) have the same colour. 
        A possible way of doing that is presented in Fig.~\ref{app-fig:grid_ex}.
        We construct a minimal Pauli set in the following way: To each party $i \in \{1,\dots, 16 \}$, we associate the colour given by the graph colouring of $G_{16}$, and later the single-qubit Pauli settings of the same colour given by the four-qubit minimal Pauli set of Fig.~\ref{fig:grid_meas} (a) (this minimal Pauli set was already present in
        {Fig.~1 of the main text}
        ).
        The resulting minimal Pauli set is given in Fig.~\ref{fig:grid_meas} (b) and we obtain $\phi_2(G_{16})=9$.
        We can easily check that the Pauli settings indeed form a Pauli set for $G_{16}$: Any two connected qubits in $G_{16}$ have different colours, and any two single-qubit Pauli settings of different colour recover all the two-body Pauli operators, as ensured by the minimal Pauli set of Fig.~\ref{fig:grid_meas} (a).
        From Eq.~\eqref{app-eq:lower_bound_g16}, it is clear that the Pauli set is minimal.

        \begin{figure}[h]
            \centering
            \includegraphics[scale=.35]{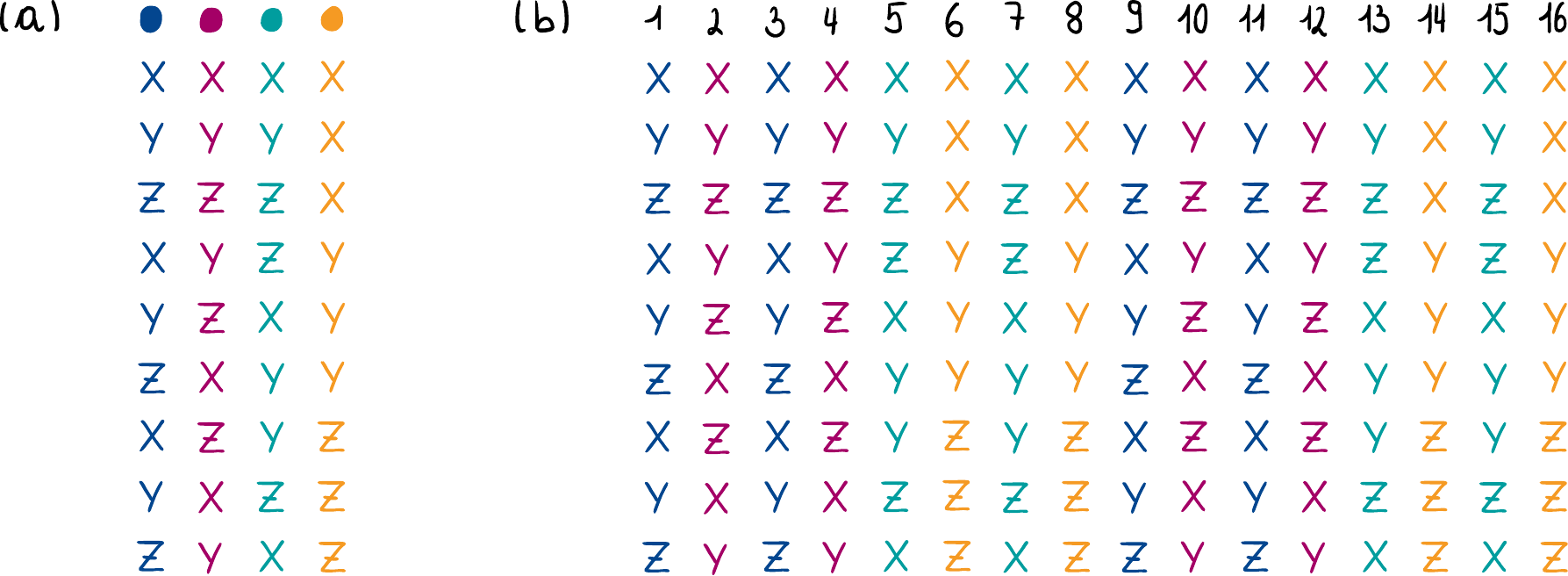}        
            \caption{%
                (a) A minimal Pauli set with nine elements for four qubits where all the two-body marginals are required. 
                It is easy to check that any two columns cover the nine two body Pauli operators.
                Each column (corresponding to settings on one qubit) is given one distinct colour.
                (b) A minimal Pauli set for $G_{16}$ of Fig.~\ref{app-fig:grid_ex}.
                The colours of the columns are given by the colours of the corresponding vertices in Fig.~\ref{app-fig:grid_ex}, and the single-qubit settings are copied accordingly from (a).
                This ensures that any two columns with distinct colours cover the nine two-body Pauli operators.
                }
            \label{fig:grid_meas}
        \end{figure}

        There exist, however, connectivity graphs for which such a construction is not possible.
        Indeed, in graph theory, the clique number $\omega(G)$ of any graph $G$ is always smaller or equal to its chromatic number $\chi(G)$.
        The clique number is given by the number of vertices in a maximal clique of $G$, and the chromatic number is the smallest number of colours needed to colour adjacent vertices with different colours.
        Graphs $G$ that fulfil $\omega(G) = \chi(G)$, such as $G$ of Fig.~\ref{app-fig:grid_ex}, are called weakly perfect (we note that the requirement for a graph $G$ to be called perfect is stronger: The clique and chromatic numbers need to be equal for any induced subgraph of $G$).
        For arbitrary connectivity graphs $G$, we thus have to consider a minimal Pauli set for $\chi(G)$ qubits, associate a distinct colour to each qubit and then proceed as described above. 
    	This results in a Pauli set, as each pair of connected qubits in $G$ has two-qubit Pauli settings that allow for the reconstruction of all the two-body Pauli expectation values.  
        The advantage of this technique is that it can drastically reduce the size of the problem.
    
        Unfortunately, this construction does not ensure minimality of the number of Pauli settings when $\omega(G) < \chi(G)$. 
        Indeed, there might be a more efficient covering of $G^{\times 3}$ than the one suggested by the above construction and thus a different Pauli set that solves the partial tomography of $G$ with less settings. 
        We summarise this statement in the following sandwich equation:
        \begin{equation}
            \phi_2({\omega(G)}) \leq \phi_2(G) \leq \phi_2({\chi(G)}).
        \end{equation}
        Recall that $\phi_k(n)$ (here $n = \omega(G)$ or $n = \chi(G)$) denotes the cardinality of minimal Pauli sets for the tomography of all $k$-body marginals of an $n$-qubit state.
        It is interesting to determine whether $\phi_2(G)$ can be strictly smaller than $\phi_2({\chi(G)})$, since it would determine the optimality of the construction.
        Trying to answer this, we considered a somewhat artificial connectivity graph with seven vertices (qubits) and edges as shown in Fig.~\ref{app-fig:g26}. 
        \begin{figure}[h]
            \centering
            \includegraphics[width=.3\linewidth]{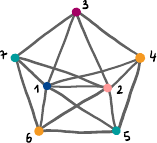}
            \caption{Connectivity graph $G_7$ for seven qubits. 
            The largest complete subgraphs of $G_7$ contain four vertices, for instance $1$, $2$, $3$, and $7$.
            This is reflected by the clique number of the graph, which is here equal to four, $\omega(G_7) = 4$.
            The chromatic number of $G_7$, which is given by the smallest possible number of colours to colour the vertices of the graph such that no connected vertices have the same colour, is equal to five, $\chi(G_7) = 5$.
            The colouring construction gives a Pauli set with $\phi_2(5) = 11$ elements, which turns out to be equal to the number of settings in the minimal Pauli set obtained by running the binary program \eqref{app-eq:mip}, i.e., $\phi_2(G_7) = \phi_2(\chi(G_7))$.}
            \label{app-fig:g26}
        \end{figure}
        
        It is the only connected graph with seven vertices whose clique number is four, whereas its chromatic number is five, which leads to $9 \leq \phi_2(G_7) \leq 11$.
        By running the binary program \eqref{app-eq:mip}, we obtain $\phi_2(G_7) = 11$, certifying that minimal Pauli sets of $G_7$ have $11$ settings.
        Similarly, we ran Eq.~\eqref{app-eq:mip} for all $26$ non-isomorphic, connected, not weakly perfect graphs with eight vertices, and did not find an instance where the colouring construction did not give a minimal Pauli set.
        In other words, there are no connectivity graphs with at most eight vertices for which $\phi_2(G) \neq \phi_2(\chi(G))$. 
        It thus remains an open question whether there are connectivity graphs for which the colouring construction does not lead to a minimal Pauli set.
    
        Moreover, since equality between the clique and chromatic numbers results in an optimal solution for the partial tomography problem, the colouring construction leads to minimal Pauli sets for many physically motivated classes of connectivity graphs.
        It is worth noting that the colouring construction is also optimal for graphs for which $\phi_2(\omega(G)) = \phi_2(\chi(G))$ despite the clique and chromatic numbers not being equal.
        For example, for $n=2, 3, 4$, $\phi_2(n) = 9$, and for $n=11, \dots, 20$, $\phi_{2}(n) = 15$ \cite{kokkala2020structure}.
        This allows us to construct minimal Pauli sets for many different classes of connectivity graphs.
        For instance, the grid example of Fig.~\ref{app-fig:grid_ex} can be extended to an arbitrary large number of qubits, and the cardinality of the minimal Pauli set remains equal to nine.
        In fact, partial tomography of any connectivity graph $G$ verifying $\chi(G) \leq 4$ can be realised with nine Pauli settings.
        This includes all planar graphs (i.e., graphs that can be drawn in the Euclidean plane in such a way that no edges cross each other, such as lines, cycles and two-dimensional lattices) \cite{appel1977solution}, as well as three-dimensional grids.
    
        The colouring construction can be generalised to higher-body marginals, that is, to $k>2$.
        Given a connectivity hyper graph $H$, a lower bound on $\phi_k(H)$ is directly put by $\phi_k(\omega(H))$, where $\omega(H)$ is the number of vertices in the largest complete subgraph of $H$.
        We recall that an $l$-vertex hyper graph with edge size $k$ is complete when its edge set is given by all $k$-subsets of $\{1,\dots,l\}$, that is $\{e \subset \{1,\dots,l\} \mid \abs{e} = k \}$ \cite{bretto2013hypergraph}.
        The natural extension of the $k=2$ case suggests to colour the vertices of $H$ such that vertices contained in the same edge have different colours.
        This is known as a strong colouring of $H$, and the smallest number of colours needed is the chromatic number of $H$, $\chi(H)$ \cite{bretto2013hypergraph}.
        Pauli sets are then constructed analogously to $k=2$: One considers a minimal Pauli set for $\chi(H)$ qubits where all the $k$-body marginals are needed, and associates a colour to each qubit. 
        Then, the Pauli set for partial tomography of $H$ is constructed by taking the single-qubit Pauli settings of the minimal Pauli set for $\chi(H)$ following the colouring of $H$, exactly as in the case of two-body partial tomography.
        Similarly, when $\phi_k(\omega(H)) = \phi_k(\chi(H))$, the Pauli set is minimal.
        If we look at the ring hypergraph $H_7$ of Fig.~\ref{app-fig:hypergraph_example}, we have $\omega(H_7) = 3$ and $\chi(H_7)=4$ and since $\phi_3(3) = \phi_3(4) = 27$, we recover that the partial tomography of $H_7$ can be performed with a minimal number of $27$ Pauli settings.
        We note that the authors of Ref.~\cite{friis2018observation} already realised that for line connectivities, marginal tomography can always be realised with $3^k$ settings, as in that case, $\chi(H) = k$.
        The case of ring connectivites is thus less trivial, as the chromatic number depends on $n$ and $k$, as discussed in the next paragraph.
    
        This construction generalises and unifies the results of Araújo \emph{et al.\@} in Ref.~\cite{araujo2022local} concerning qubits.
        First, a general construction is proposed for connectivity hypergraphs where the vertices are ordered in a lattice, and where the hyperedges have a periodic structure. 
        However, because of the generality of their construction, it is argued that the number of Pauli settings can be reduced by looking at specific cases.
        A similar idea to the graph colouring is also introduced, however, not connected to smaller instances.
        Concerning $k=2$, a general construction leading to the costly number of settings of $3^{\chi(G)}$ is presented.
        Concerning higher $k$s, a few cases of hypergraphs that can be coloured using only $k$ colours (the size of the hyperedges) that is, $\chi(H) = k$, are discussed.
        This is recovered by our construction, and we add that $k$-body overlapping tomography of a connectivity hypergraph $H$ can be done with $3^k$ Pauli settings if $\chi(H) \leq k+1$.
        This comes from the fact that minimal Pauli sets for $k$-body overlapping tomography of $k+1$ qubits have $3^k$ elements \cite{torres2013survey}.
        Finally, cyclic topologies, such as cycles and toruses, are discussed.
        There again, the number of Pauli settings can be improved using our colouring construction. 
        For instance, it is know that cycle hypergraphs with hyperedge size $k$ have a chromatic number equal to $k + \lceil \nicefrac{r}{q} \rceil$ where $q$ is the quotient and $r$ the remainder of $n$ divided by $k$ (see Theorem 3.1 of Ref.~\cite{pk2019distance}).
        So, when $r \leq q$ (which is in particular satisfied when $n \geq k^2-1$), the chromatic number is at most $k+1$ and thus minimal Pauli sets have $3^k$ elements, for which the method described in Ref.~\cite{araujo2022local} needs twice as many settings.
        The case of $H_7$ is again recovered, as its chromatic number is four, which is equal to $k+1$.
        The strong colouring of $H_7$ is shown in Fig.~\ref{app-fig:hypergraph_example}.
    
        We recently became aware of similar constructions in the context of sparse Pauli-Lindblad noise models \cite{berg2023techniques,jaloveckas2023efficient}.
        The authors of Ref.~\cite{berg2023techniques} consider a similar problem of finding $n$-qubit Pauli operators that cover some lower order Pauli operators. 
        They independently notice that this problem corresponds to covering arrays, and present a similar construction based of vertex colouring.
        In Ref.~\cite{jaloveckas2023efficient}, the authors present a construction to recover all two-body marginals of $n$-qubit systems that requires $3(1+2 \lceil \log_2(n-2) \rceil$ Pauli settings.
        We also note that Theorem 2 of Ref.~\cite{jaloveckas2023efficient} states that for connectivity graphs $G$ with $\omega(G) \leq 4$, the number of elements in minimal Pauli sets is nine.
        We have shown that this is not true, and refer to $G_7$ of Fig.~\ref{app-fig:g26} as a counterexample.

    \subsection{On the optimality of constructions for $k=2$} \label{app-sec:recursion}
        Unfortunately, in the case of complete connectivity graphs, the colouring construction does not reduce the size of the problem. 
        In this case, it can be convenient to resolve to explicit constructions, which might come at the cost of optimality. 
        In Ref.~\cite{garciaperez2020pairwise}, the authors propose a method to construct Pauli sets for two-body overlapping tomography that translates to an upper bound on $\phi_2$, 
        \begin{equation} \label{app-eq:garcia_scaling}
            \phi_2(n) \leq 6 \lceil \log_3(n) \rceil +3 .
        \end{equation}
        Alternatively, a well-known recursive construction for covering arrays shows that, from Pauli sets for $n_1$ and $n_2$ qubits, it is possible to build a Pauli set for $n_1n_2$ qubits \cite{colbourn2004combinatorial, torres2013survey}.
        We present here a slightly modified version that requires one less setting, and show that the recursive construction also leads to an upper bound on $\phi_2$.
        
        The recursive construction exemplified in Fig.~\ref{app-fig:recursive}, and goes as follows.
        \begin{figure}
            \centering
            \includegraphics[scale=.35]{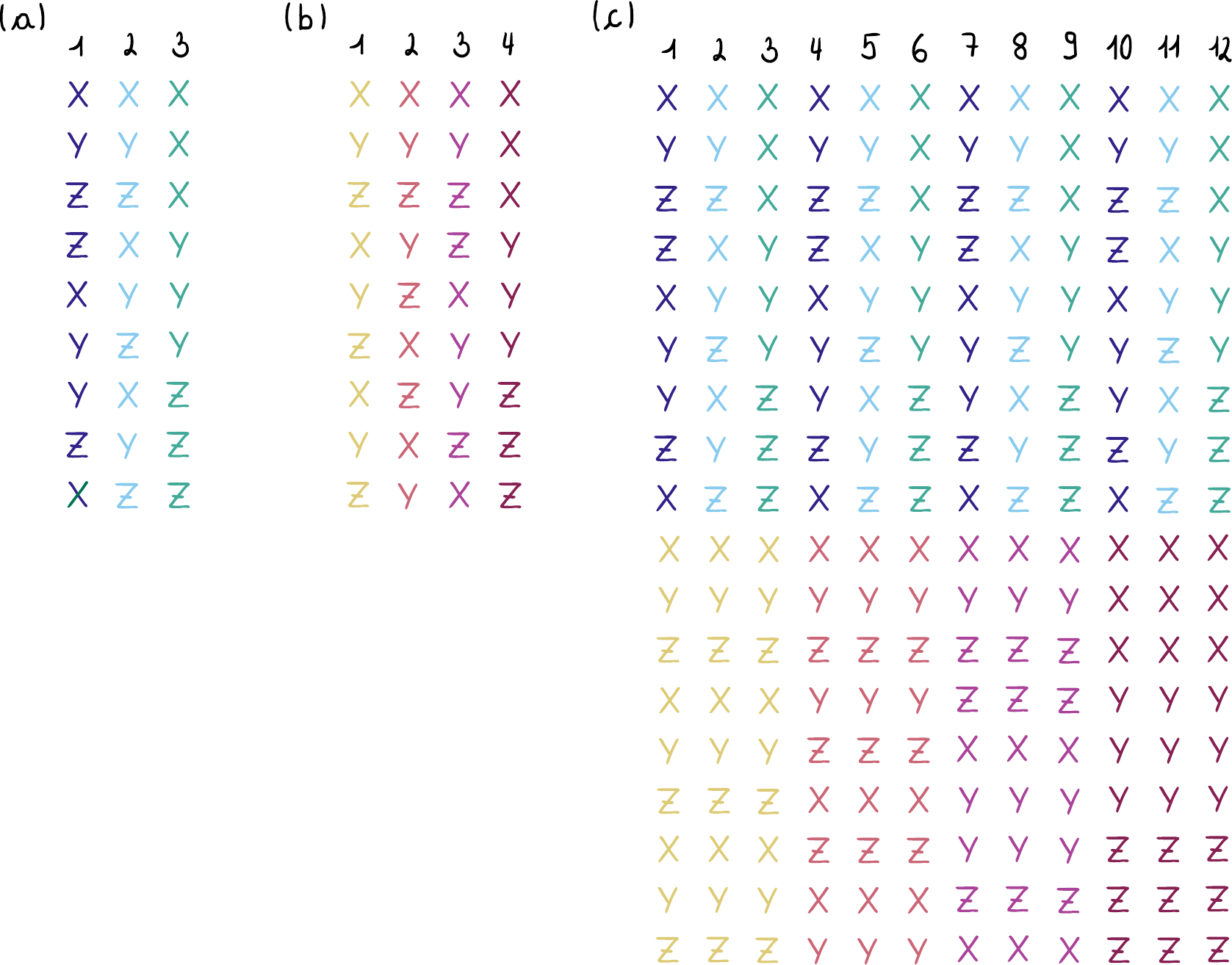}
            \caption{%
            Recursive construction of a Pauli set. 
            (a), (b) Minimal Pauli set on $n_1=3$ and $n_2=4$ qubits respectively for reconstructing all two-body marginals. 
            (c) Recursively constructed Pauli set on $3\times 4=12$ qubits, where the colours correspond to the single-qubit settings of (a) and (b). 
            From the first nine measurements, all two-body marginals can be obtained except for pairs which have the same colour, such as for instance the qubit pair $\{1,4\}$. 
            Similarly, the marginal of the qubit pair $\{1,2\}$ cannot be reconstructed from the last nine measurements alone, as the settings only cover $XX$, $YY$, and $ZZ$.
            However, when considering all $18$ Pauli settings, any qubit pair is covered by the nine two-body Pauli operators.
            Indeed, for each qubit pair either the upper or the lower nine settings have different colours, therefore covering the nine combinations $XX, XY, XZ, YX, YY, YZ, ZX, ZY ,ZZ$.}
            \label{app-fig:recursive}
        \end{figure}
        First, assume we know (not necessarily minimal) Pauli sets for $n_1$ and $n_2$ qubits, with sizes $m_1 \geq \phi_2(n_1)$ and $m_2\geq \phi_2(n_2)$ respectively. 
        Write them as $\mathcal{M}_1 = \set{ M_{\alpha}^1 }_{\alpha=1}^{m_1}$ and $\mathcal{M}_2 = \set{ M_{\alpha}^2 }_{\alpha=1}^{m_2}$. 
        Therein, $M_\alpha^\ell$ is an $n_\ell$-qubit Pauli operator for $\ell=1,2$ and $\alpha \in \{1,\dots, m_\ell\}$.
        Without loss of generality, assume that $X^{\otimes n_\ell}$ is part of both Pauli sets ($\ell =1, 2$).
        Then, take the Pauli settings $\set{(M_{\alpha}^1)^{\otimes n_2}}_{\alpha=1}^{m_1}$, which amounts to $m_1$ Pauli settings acting on $n_1n_2$ qubits. 
        It is easy to see that all the two-body marginals can be obtained for all pairs of qubits, except for the pairs $(x n_1 + z, y n_1 + z)$ for $x,y = 0, \dots, n_2-1$, provided $x<y$, and $z=1,\dots, n_1$. 
        To amend this, we complete the set of Pauli settings with $m_2$ additional operators of the form
        \begin{equation}
            \begin{split}
                &(M_{\alpha}^2)_{1, n_1+1, 2n_1+1,\ldots} \otimes (M_\alpha^2)_{2, n_1+2, 2n_1+2,\ldots} \quad\otimes \ldots \otimes (M_\alpha^2)_{n_1, 2n_1, 3n_1,\ldots}, \quad \alpha = 1 ,\ldots, m_2,
            \end{split}
        \end{equation}
        where the indices indicate on which qubits the Pauli operators in each $M_{\alpha}^2$ should act. 
        In Fig.~\ref{app-fig:recursive} the construction is demonstrated for the case $n_1=3$ and $n_2=4$, leading to a new Pauli set for $12$ qubits.
        Recall that $X^{\otimes n_\ell}$ ($\ell=1,2$) appears in $\mathcal{M}_1$ and $\mathcal{M}_2$ respectively, thus twice in the Pauli set for $n_1n_2$ qubits. 
        We end up with a Pauli set $\mathcal{M}_{1\times2}$ for $n_1n_2$ qubits containing $m_1+m_2-1$ measurement settings. 
        Clearly, $\phi_2(n_1n_2) \leq m_1+m_2-1$.
        This holds in particular when $\mathcal{M}_1$ and $\mathcal{M}_2$ are minimal Pauli sets, we obtain
        \begin{equation}
            \phi_2(n_1n_2) \leq \phi_2(n_1)+\phi_2(n_2)-1.
            \label{eq:recursion-bound-n1n2}
        \end{equation}

        To obtain an upper bound for the scaling of $\phi_2(n)$ we set $n_1= n = \alpha^x$, $n_2 = \alpha$ and define $\xi(x) = \phi_2(\alpha^x)$, rewriting Eq.~\eqref{eq:recursion-bound-n1n2} as $\xi(x+1) \leq \phi_2(\alpha) + \xi(x) - 1$. 
        Using $\xi(1) = \phi_2(\alpha)$, we obtain $\xi(x) \leq (x - 1) \left[ \phi_2(\alpha) - 1 \right] + \phi_2(\alpha)$ by recurrence, which holds for $x \in \mathbb{N}_0$. 
        Rearranging the terms we get $\phi_2(\alpha^x) \leq x( \phi_2(\alpha) - 1) + 1$, or rather, 
        \begin{equation} \label{app-eq:myscaling}
            \phi_2(n) \leq \left[ \phi_2(\alpha) - 1 \right] \lceil \log_\alpha(n) \rceil + 1, \quad \forall \alpha \geq 2.
        \end{equation}
        In Fig.~\ref{fig:scaling}, we compare this scaling for different values of $\alpha$ for which we know $\phi_2$ exactly, to the scaling of Ref.~\cite{garciaperez2020pairwise} given in Eq.~\eqref{app-eq:garcia_scaling}.
        \begin{figure}
            \centering
            \includegraphics[width=.9\linewidth]{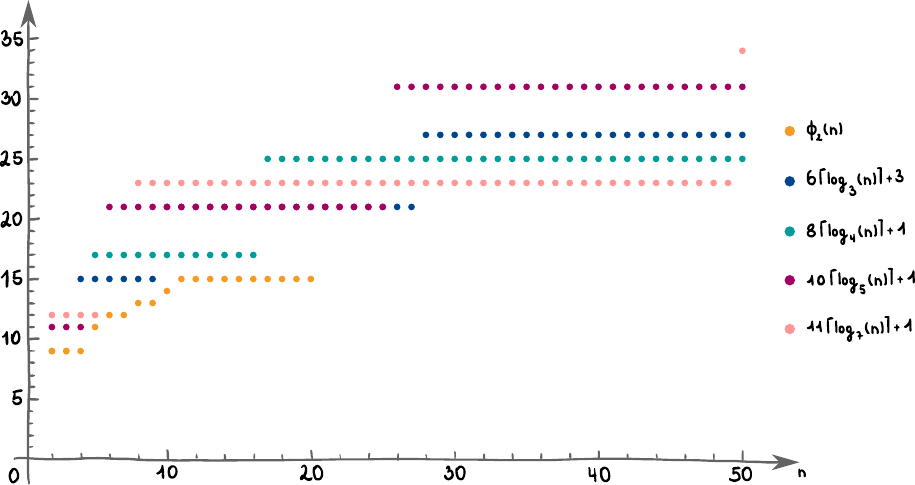}
            \caption{Number of Pauli settings as a function of number of qubits $n$, for different constructions. 
            \textcolor{orange}{Orange:} cardinality of minimal Pauli sets, $\phi_2(n)$ \cite{kokkala2020structure}. 
            These values correspond to the exact minimal number of Pauli settings needed to reconstruct all two-body marginals.
            As stressed before, they are only known up to $20$ qubits.
            \textcolor{myblue}{Blue:} construction from Ref.~\cite{garciaperez2020pairwise}, given by Eq.~\eqref{app-eq:garcia_scaling}. 
            \textcolor{myteal}{Green:} Equation \eqref{app-eq:myscaling} with $\alpha = 4$, where Eq.~\eqref{app-eq:myscaling} corresponds to the number of settings arising from the covering array construction explained in this appendix.
            \textcolor{mypurple}{Purple:} Equation \eqref{app-eq:myscaling} with $\alpha = 5$. 
            \textcolor{mypink}{Pink:} Equation \eqref{app-eq:myscaling} with $\alpha = 7$.}
            \label{fig:scaling}
        \end{figure}
        So, for example, while the construction in Ref.~\cite{garciaperez2020pairwise} is elegant and does not depend on smaller Pauli sets, it requires $6 \lceil \log_3(12) \rceil +3 = 21$ measurements to obtain all two-body marginals of a $12$-qubit system, while the recursive construction only requires $\phi_2(3)+\phi_2(4)-1=17$ settings. 
        Notice that none of the constructions are optimal, since it is known that $\phi_2(12) = 15$ \cite{kokkala2020structure}.
    
        Similar recursive constructions exist for $k>2$, for instance, it can be shown that \cite{colbourn2004combinatorial, torres2013survey}
        \begin{equation}
            \phi_3(2n) \leq \phi_3(n) + 2 \phi_2(n).
        \end{equation}
        For the explicit construction and for larger $k$s, we refer the reader to Refs.~\cite{colbourn2004combinatorial, torres2013survey} and to the references therein.
    
        Lastly, we note that significant effort has been deployed by the combinatorial designs community to obtain small covering arrays, which directly translates to Pauli sets.
        For readers interested by the smallest Pauli sets known up to date for a given number of qubits and a given $k$, we refer to the online tables in Refs.~\cite{colbourncovering, torresuniform}. 

    \subsection{Minimal Pauli sets} \label{app-sec:min_pauli_sets}
        From the covering array literature, we can obtain $k=2$ minimal Pauli sets for $n \leq 20$ \cite{kokkala2020structure} and for $k=3$, up to $n=6$ qubits \cite{colbourn2004combinatorial, torres2013survey}.
        This is summarised in the following table.
        \begin{equation}
        \begin{tabular}{|c|ccccccc|}
				\hline 
				$n$ & $4$ & $5$ & $6$ & $7$  & $9$ & $10$ & $20$ \\ \hline
				$\phi_2(n)$ & $9$ & $11$ & $12$ & $12$ & $13$ & $14$ & $15$ \\
				$\phi_3(n)$ & $27$ & $33$ & $33$ &      &      &      &      \\ \hline
			\end{tabular}
        \end{equation}
        For integers $n^\prime \in \{1,\dots,20 \}$ that are not displayed, $\phi_k(n^\prime) = \phi_k(n)$, where $n$ is the closest larger integer to $n^\prime$ displayed.
        See also Fig.~\ref{fig:scaling} for $k=2$.
        
        The minimal Pauli set used for the six-photon experimental demonstration presented in the main text was obtained through the binary program described in Eq.~\eqref{app-eq:mip}.
        The output is given by 
        \begin{equation}
            \includegraphics[scale=.35]{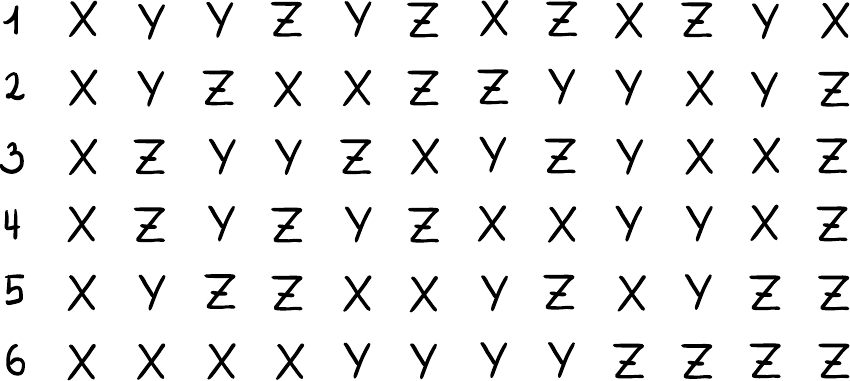}
        \end{equation}
        where the columns represent the six-qubit settings.
    
        For three-body overlapping tomography, $27$ Pauli settings are necessary.
        For three- and four-qubit systems, this is also sufficient and the minimal Pauli set obtained using Eq.~\eqref{app-eq:mip} reads
        \begin{equation}
            \includegraphics[scale=.35]{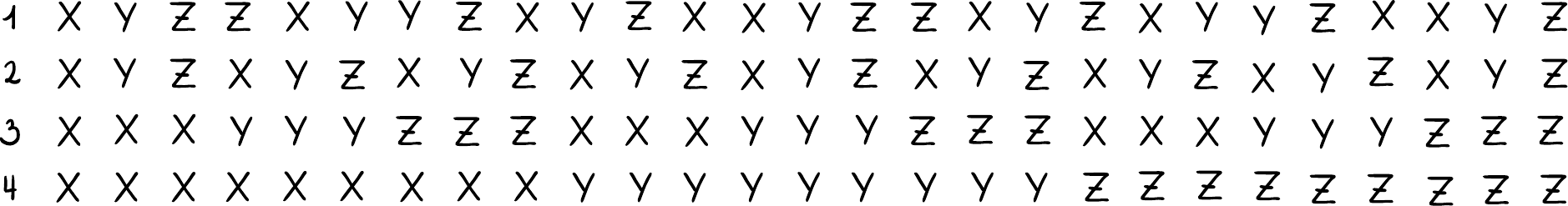}
        \end{equation}
        where four-qubit settings are given by the columns.
        Moreover, Eq.~\eqref{app-eq:mip} certifies that $\phi_3(5) = \phi_3(6) = 33$.
        We stress again that, to the best of our knowledge, the fact that this covering array is minimal for $n=5$ was not previously known.
        For $n=6$, the result can be found in Ref.~\cite{colbourn2004combinatorial}.
        The minimal Pauli set for three-body overlapping tomography of six qubits obtained through the hypergraph generalisation of Eq.~\eqref{app-eq:mip} is given by
        \begin{equation}
            \includegraphics[scale=.35]{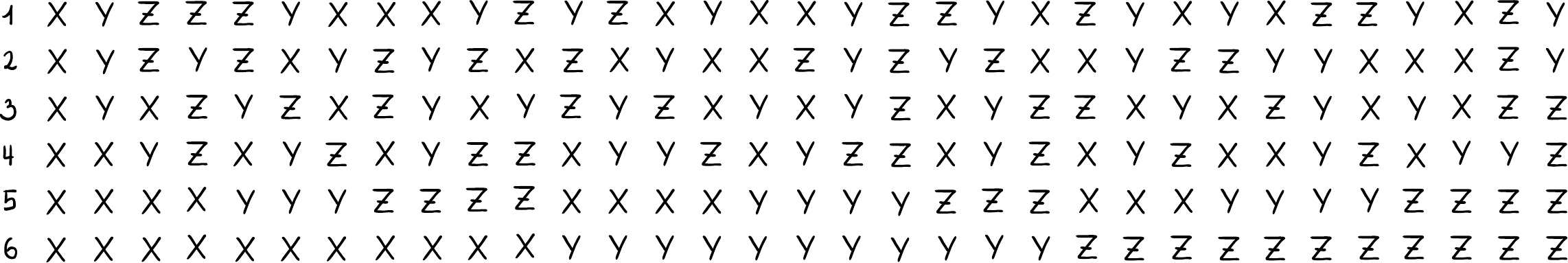}
        \end{equation}
        with columns representing the six-qubit Pauli settings.
        A minimal Pauli set for three-body overlapping tomography of a five-qubit system can straightforwardly be obtained by removing one of the rows of the six-qubit minimal Pauli set. 
    
        Finally, we present the minimal Pauli set for three-body overlapping tomography of seven qubits in the ring structure $H_7$ (see Fig.~\ref{app-fig:hypergraph_example}) obtained through the hypergraph generalisation of Eq.~\eqref{app-eq:mip},
        \begin{equation}
            \includegraphics[scale=.35]{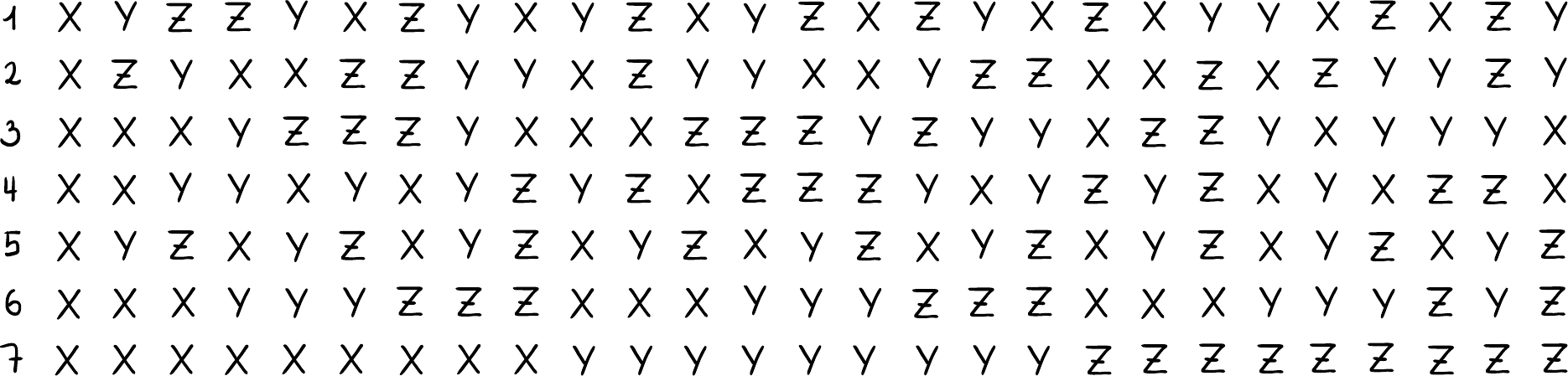}
        \end{equation}
        where the seven-qubit Pauli settings are given by the columns.

\section{Optimal tomography beyond Pauli measurements}\label{app-sec:min_settings}
    In the main text, we claim that choosing $3^k$ three-dimensional real vectors $\vec{v}_\alpha^{(i)}$ as measurement directions for every of the $n$ parties leads to tomographically complete measurement settings 
    \begin{equation}  \label{app-eq:meas_settings}
        \meas_\alpha =  \bigotimes_{i=1}^n \vec{v}_\alpha^{(i)} , \quad \alpha = 1,\dots,3^k,
    \end{equation}
    if the vectors are linearly independent, that is
    \begin{equation}
        \mathrm{span} \left( \left\{\bigotimes_{i \in \pairset} \vec{v}_\alpha^{(i)}\right\}_{\alpha =1}^{3^k}  \right) = \mathbb{R}^{3^k}
    \end{equation}
    for all $\pairset \subset \{1,\dots,n\}$ of cardinality $k$.

    Here, we show that this condition on the vectors can be fulfilled and, moreover, is fulfilled almost surely if the vectors are chosen randomly.
    This means that the vectors are independently and identically distributed (i.i.d.) with respect to the uniform distribution on the unit sphere in $\mathbb{R}^3$. 
    
    Moreover, we do not need to restrict ourselves to only the qubit case and therefore to $\mathbb{R}^3$. 
    We can consider systems of $n$ qudits of local dimension $d$ and take an orthogonal basis of the $m:= d^2-1$ dimensional space of operators instead of the Pauli operators.  
    Indeed, the following arguments hold for any combination of $k$ randomly chosen $m$-dimensional vectors $\vec{v}_\alpha^{(i)}$.
    
    To show this, we first provide a lemma that, intuitively, states that the only vector which is orthogonal to randomly chosen product vectors is the null vector.
    \begin{lemma}\label{lem:orthtens}
        Let $\vec{v}^{(i)} \in\mathbb{R}^m$ for $i=1,\dots,k$ be vectors which are i.i.d.\@ with respect to the uniform distribution on the unit sphere in $\mathbb{R}^m$, and $\vec{T}\in\mathbb{R}^{m^k}$ an arbitrary vector. 
        It holds that
        \begin{equation}\label{app-eq:orthT}
            \mathbb{P}\left[ \vec{T} \times \bigotimes_{i=1}^k \vec{v}^{(i)} = 0 \right] = 0
        \end{equation}
        for $\vec{T} \neq \vec{0}$. 
        So, the product vector $\bigotimes_{i=1}^k \vec{v}^{(i)}$ is almost surely not orthogonal to $\vec{T}$.
    \end{lemma}
    Before we prove this lemma, we want to give a few remarks on the notation and the structure of the probability space. 
    First, we denote the unit sphere in $\mathbb{R}^m$ by $\mathcal{S}^{m-1}$ and the probability measure of the uniform distribution on the sphere by $\mu_m$. 
    The event which we consider in Eq.~\eqref{app-eq:orthT} lies in the space $\mathcal{S}^{m-1} \times \dots \times \mathcal{S}^{m-1}$, and is given by the set $M = \lbrace (\vec{v}^{(1)}, \dots, \vec{v}^{(k)}) \mid \vec{T} \times \bigotimes_{i=1}^k \vec{v}^{(i)} = 0 \rbrace$ and the probability measure is the product measure $\mathbb{P}=\mu_m \otimes \dots \otimes \mu_m$. Lastly, we denote with $\mathds{1}_A(x)$ the characteristic function of $A$. This function attains the value $1$ if and only if $x\in A$, and $0$ otherwise. 

    With these points we can now prove the lemma.
    \begin{proof}
        The main idea of the proof is to calculate the probability \eqref{app-eq:orthT} by decomposing it into $k$ integrals over the $k$ unit spheres and integrating iteratively over each single sphere. 
        To achieve this, we define for every $j$ the set $M_j :=\lbrace (\vec{v}^{(1)}, \dots, \vec{v}^{(j)}) \mid \vec{X} \times \bigotimes_{i=1}^k \vec{v}^{(i)} = 0 \ \forall \vec{v}^{(j+1)},\dots, \vec{v}^{(k)} \in \mathcal{S}^{m-1} \rbrace$. 
        This set contains exactly these combinations $(\vec{v}^{(1)}, \dots, \vec{v}^{(j)})$ such that the tensor product of these specific vectors and all possible vectors from the last $k-j$ unit spheres is orthogonal to $\vec{T}$ and therefore contained in the set $M$. 
        Note that the set $M_k = M$, and $M_j=\mathcal{S}^{m-1}$ for all $j\in\{1,\dots, k\}$ if $\vec{T}=\vec{0}$.
        We will show that in the case $\vec{T} \neq \vec{0}$ at least one of these sets is of zero measure.

        As a further prerequisite we consider for a fixed combination $(\vec{v}^{(1)}, \dots, \vec{v}^{(j-1)})$ of the first $j-1$ vectors the set $N_j = \{ \vec{v}^{(j)} \mid (\vec{v}^{(1)}, \dots, \vec{v}^{(j)}) \in M_j \}$ which contains all vectors in the $j$-th unit sphere that lead to a combination contained in $M_j$. 
        Using this definition, we can calculate the probability 
        \begin{align}
            \mu_m[N_j] 
            := \mu_m \left[ \{ \vec{v}^{(j)} \mid ( \vec{v}^{(1)}, \dots, \vec{v}^{(j)} ) \in M_j \} \right] 
            = \int_{\mathcal{S}^{m-1}} \mathds{1}_{M_j}( \vec{v}^{(1)}, \dots, \vec{v}^{(j)} )\  \mathrm{d}\mu_m(\vec{v}^{(j)}) .
        \end{align}
        We argue now that there are only two possible cases: Either $N_j$ contains the whole unit sphere $N_j= \mathcal{S}^{m-1}$, or the probability $\mu_m[N_j]$ vanishes: $\mu_m[N_j] = 0$. For this let us assume that $\mu_m[N_j] >0$. Since every proper subspace of $\mathbb{R}^m$ has measure $0$, this is possible only if a basis $\{ \vec{e}_i \}_{i=1}^m$ of $\mathbb{R}^m$ lies in $N_j$. Then every vector $\vec{v}^{(j)}$ can be decomposed in this basis and is thus also contained in $N_j$. We therefore have $N_j=\mathcal{S}^{m-1}$ and $\mu_m[N_j] =1$. We can therefore write the above probability as
        \begin{align}\label{app-eq:probj}
            \mu_m[N_j] 
            = \int_{\mathcal{S}^{m-1}} \mathds{1}_{M_j}( \vec{v}^{(1)}, \dots, \vec{v}^{(j)} )\  \mathrm{d}\mu_m(\vec{v}^{(j)}) 
            = \mathds{1}_{M_{j-1}}( \vec{v}^{(1)}, \dots, \vec{v}^{(j-1)}  ).
        \end{align}
        We can now iteratively calculate the probability in Eq.~\eqref{app-eq:orthT} using Tonelli's Theorem and Eq.~\eqref{app-eq:probj}:
        \begin{align*}
            \mathbb{P}[M]
            &= \mathbb{P}\left[ \vec{T} \times \bigotimes_{i=1}^k \vec{v}^{(i)} = 0 \right] \\
            & = \int_{\mathcal{S}^{m-1} \times \cdots \times \mathcal{S}^{m-1}} \mathds{1}_M (\vec{v}^{(1)}, \dots , \vec{v}^{(k)}) \  \mathrm{d}\mathbb{P}(\vec{v}^{(1)}, \dots , \vec{v}^{(k)}) \\
            & = \int_{\mathcal{S}^{m-1}} \cdots \int_{\mathcal{S}^{m-1}} \mathds{1}_M (\vec{v}^{(1)}, \dots , \vec{v}^{(k)})
            \  \mathrm{d}\mu_m(\vec{v}^{(k)}) \dots \  \mathrm{d}\mu_m(\vec{v}^{(1)}) & \mathrm{(Tonelli)}\\
            & = \int_{\mathcal{S}^{m-1}} \cdots \int_{\mathcal{S}^{m-1}} \mathds{1}_{M_{k-1}} (\vec{v}^{(1)}, \dots , \vec{v}^{(k-1)})
            \  \mathrm{d}\mu_m(\vec{v}^{(k-1)}) \dots \  \mathrm{d}\mu_m(\vec{v}^{(1)}) & \mathrm{(Eq.\ \eqref{app-eq:probj})} \\
            &= \dots \\
            &= \int_{\mathcal{S}^{m-1}} 
            \mathds{1}_{M_{1}} (\vec{v}^{(1)})
            \  \mathrm{d}\mu_m(\vec{v}^{(1)})  \\
            &= \mu_m [M_1] \\
            &= \mu_m [N_1].
        \end{align*}

        By the same reasoning as before, either the probability in the last line vanishes and it holds $\mathbb{P}[M] = 0$, or the set $N_1$ contains the whole unit sphere $\mathcal{S}^{m-1}$ and it holds $\mathbb{P}[M] = 1$. However, this is equivalent to  the vector $\vec{t}$ being orthogonal on every possible combination $\bigotimes_{i=1}^k \vec{v}^{(i)}$ of unit vectors. The only vector $\vec{T}$ fulfilling this is the null vector $\vec{0}$. This concludes the proof.
    \end{proof}
    Using this lemma, we can directly prove the wanted statement.
    \begin{theorem}
        Let $\vec{v}_\alpha^{(i)} \in \mathbb{R}^m$ for $i=1,\dots,k$ and $\alpha=1,\dots,m^k$ be i.i.d.\@ vectors with respect to the uniform distribution $\mu_m$ on the unit sphere $\mathcal{S}^{m-1}$ in $\mathbb{R}^m$. Then it holds almost surely that
        \begin{equation}\label{eq:spannk}
            \mathrm{span} \left( \left\{\bigotimes_{i =1}^k \vec{v}_\alpha^{(i)}\right\}_{\alpha =1}^{m^k}  \right) = \mathbb{R}^{m^k},
        \end{equation}
        i.e., the tensor products of these vectors are linearly independent and form a basis of $\mathbb{R}^{m^k}$.
    \end{theorem}
    \begin{proof}
        Notice that Eq.~\eqref{eq:spannk} is equivalent to the fact that the matrix $M$ which contains as columns the tensor products $\bigotimes_{i =1}^k \vec{v}_\alpha^{(i)}$ has a non-vanishing determinant. 
        To prove that this condition is almost surely fulfilled, we have to show $\mathbb{P}[\det(M) = 0] =0$.

        First, we decompose the determinant of $M$ into two independent parts by applying the Laplace expansion along the first column.  
        Reordering the resulting terms leads to 
        \begin{equation}
            \det(M) = \left( \bigotimes_{i =1}^k \vec{v}_1^{(i)} \right) \times \vec{T_1},
        \end{equation}
        where the vector $\vec{T_1}$ fulfils $(\vec{T_1})_i = \pm \det(M_{i,1})$, and the $M_{i,j}$ denote the submatrices from $M$ resulting from removing the $i$th row and the $j$th column. 
        This decomposition also follows from the fact that the determinant is linear in the first column. 
        It is important to note that the vector $\vec{T_1}$ does not depend on the vectors $\vec{v}_1^{(i)}$, $i=1,\dots k$, so the two parts in the above product are statistically independent. 

        One of the entries in $\vec{T_1}$ is the determinant $\det(M_{1,1})$, where the first row and column were removed from $M$. This determinant can again be written as 
        \begin{equation}
            \det(M_{1,1}) = \left( \bigotimes_{i =1}^k \vec{v}_2^{(i)} \right) \times \vec{T_2}.
        \end{equation}
        Here, $\vec{T_2}$ contains as first entry a $0$ and is independent from the vector $\bigotimes_{i =1}^k \vec{v}_2^{(i)}$. 
        Again, one entry is given by the determinant of the submatrix $M_{\{1,2\},\{1,2\}}$, where the first and second rows and columns were removed. 
        Iterating this scheme we can always decompose the determinant $M_{\{1,...,m\},\{1,...,m\}}$ into two independent parts $\left( \bigotimes_{i =1}^k \vec{v}_{m+1}^{(i)} \right)$ and $\vec{T}_{m+1}$, where one entry of the vector $\vec{T}_{m+1}$ is given by $M_{\{1,...,m+1\},\{1,...,m+1\}}$ for $m=0,...,n-1$. 
        Also, the number of entries in $\vec{T}_{m}$ which are $0$ increases in every step by $1$. 
        In the last step we obtain the vector $\vec{T}_{n}$ where only one single entry equals $1$ and otherwise the entries are $0$.

        We can now directly calculate the probability $\mathbb{P}[\det(M)=0]$ using this decomposition as
        \begin{equation}
            \mathbb{P}[\det(M)=0] = \mathbb{P}[ \left( \bigotimes_{i =1}^k \vec{v}_1^{(i)} \right) \times \vec{T_1} = 0].
        \end{equation}
        From Lemma \ref{lem:orthtens}, we can conclude that the right hand side vanishes if and only if $\vec{T_1}\neq\vec{0}$. 
        We therefore obtain 
        \begin{equation}
            \mathbb{P}[\det(M)=0] = \mathbb{P}[\vec{T_1} = \vec{0}].
        \end{equation}
        The right hand side is clearly upper bounded by the probability that only one special entry of $\vec{T_1}$, namely $\det(M_{1,1})$, equals zero. Repeating the above argument and using the iterative scheme described before, we get
        \begin{equation}
        \begin{split}
            0 &\leq \mathbb{P}[\det(M)=0] = \mathbb{P}[\vec{T_1} = \vec{0}] \\
            &\leq \mathbb{P}[\det(M_{1,1}) = 0] = \mathbb{P}[\vec{T_2} = \vec{0}] \\
            &\leq \dots  \leq \mathbb{P}[\vec{T_n} = \vec{0}].
        \end{split}
        \end{equation}
        Since $\vec{T_n}$ is a fixed vector with one entry equals $1$, the probability that it is equal to $\vec{0}$ vanishes and therefore it holds $\mathbb{P}[\det(M)=0] = 0$, which concludes the proof.
    \end{proof}

\section{Confidence regions and numerical optimisation} \label{app-sec:conf_regs_and_num_opt}
    The goal of quantum state tomography is to obtain an estimate $\hat{\varrho}$ for an unknown density operator $\varrho$. Typically, measurement data is obtained from independent samples of $\varrho$ which are measured using a tomographically complete set of measurements.
    Since only a finite number of samples $N$ can be measured, the estimate $\hat{\varrho}$ invariably differs from the true state $\varrho$. Therefore, to have a trustable estimate, it is important to give rigorous guarantees on the maximum distance between $\hat{\varrho}$ and $\varrho$. This can be done by means of confidence regions (or credible regions, if a Bayesian approach is used).
    
    In Ref.~\cite{degois2023userfriendly}, it is shown that with $N$ samples of the state $\varrho$ one can reconstruct a Hermitian operator $\hat{\varrho} = M^+ \vec{f}$, where $\vec{f}$ is the measurement data and $M^+$ a linear map related to the measurements (see below), such that
    \begin{equation} \label{app-eq:conf_reg}
        \Pr[\norm{\hat{\varrho} - \varrho}_2 < \eps \sigma] \geq 1 - \delta.
    \end{equation}
    Here, $1 - \delta \in [0,1]$ is the confidence level, $\eps = 3\sqrt{u}(\sqrt u+\sqrt{u+1})$, with $u = \nicefrac{2}{9N}\log(8/\delta)$, the distance $\norm{\cdot}_2$ is measured in the Hilbert-Schmidt norm, and the parameter $\sigma$ is a function of the measurement settings and related to the variance in the measurement results. 
    The resulting region $\norm{\hat{\varrho} - \varrho}_2 < \eps \sigma$ is a sphere in the space of Hermitian operators around $\hat{\varrho}$ and is such that $\varrho$ is inside it with probability at least $1 - \delta$.

    Notice that $\sigma$ is the only parameter associated with the measurement settings, and that ideally we want it to be as small as possible. It can be computed in the following way. Suppose we have a set of observables $\{\meas_\alpha\}_\alpha$ that we want to measure.
    Each $\meas_\alpha$ can be decomposed into its measurement effects $\{\Pi_\alpha^o\}_o$, with $\{o\}$ labelling the possible measurement outcomes of $\meas_\alpha$.
    When a state $\varrho$ is measured according to $\mathcal{M}_\alpha$, outcome $o$ occurs with probability $\tilde{p}_{o \mid \alpha}= \tr\left(\Pi^o_\alpha \varrho\right)$, and we suppose that $\mathcal{M}_\alpha$ is measured a total of $N_\alpha$ times.
    We rescale the projectors and collect them into a single positive operator-valued measure (POVM), 
    \begin{equation}
        \Pi = \left\{ \frac{N_\alpha}{N} \ \Pi^o_\alpha \right\}_{o,\alpha}, \quad \sum_\alpha N_\alpha = N, 
    \end{equation}
    and label the measurement outcome associated to the effect $\frac{N_\alpha}{N} \ \Pi^o_\alpha$ by $(o,\alpha)$. The outcome $(o,\alpha)$ is then obtained with probability
    \begin{equation}
        p_{(o, \alpha)} = \tr\left( \frac{N_\alpha}{N} \ \Pi^o_\alpha \varrho \right) = \frac{N_\alpha}{N} \tilde{p}_{o \mid \alpha} ,
    \end{equation}
    which we collect into a single vector $\vec{p}$. For our purposes, which is to get a quantifier of the quality of a measurement set, we assume that performing each projective measurement $\meas_\alpha$ a certain number $N_\alpha$ of times is equivalent to performing the single generalised measurement $\Pi$.

    We now implement the tomographic experiment: We perform $N$ times the POVM $\Pi$, and obtain the outcome $(o,\alpha)$ a certain number $N_{(o,\alpha)}$ of times. We build the frequency vector $\vec{f}$ with entries 
    \begin{equation}
        f_{(o,\alpha)} = \frac{N_{(o,\alpha)}}{N}.
    \end{equation}
    Of course, 
    \begin{equation}
        \lim_{N \rightarrow \infty} \frac{N_{(o,\alpha)}}{N} = p_{(o,\alpha)}.
    \end{equation}
    This procedure can be summarised by the measurement map 
    \begin{equation} \label{app-eq:meas_map}
        M: \mathbb{C}^{d^2} \rightarrow \mathbb{R}^{\abs{\Pi}} : \varrho \mapsto M \varrho  = \vec{p},
    \end{equation}
    with $\abs{\Pi}$ being the number of possible outcomes. To represent $M$ as a matrix, we can take its rows to be vectorisations of the elements of $\Pi$.
    Since we assume the measurements to be tomographically complete, then $M$ has a left-inverse (which we take to be the pseudo-inverse $M^+$) such that $M^+ \vec{p} = \varrho$ and $\hat{\varrho} = M^+ \vec{f}$. 
    From this map, we finally compute $\sigma$ as
    \begin{equation} \label{app-eq:sigma}
        \sigma = \max_k \norm{M^+_k}, 
    \end{equation}
    where $M_k^+$ is the $k$th column vector of $M^+$ (see Ref.~\cite{degois2023userfriendly} for details).

    In the following, we use $\sigma$ as a figure of merit to optimise the measurement directions for the six qubits overlapping tomography experiment. Thus, we aim at finding $9 \times 6 = 54$ local measurement directions that lead to small confidence regions for each pair of qubits, that is, to small $\sigma_\pairset$, for all $\pairset \subset \{1,\dots,6\}$ with $\abs{\pairset} = 2$.
    We focus our attention on $\max_\pairset \sigma_\pairset = \sigma_{\max}$, such that we can phrase Eq.~\eqref{app-eq:conf_reg} as
    \begin{equation} \label{app-eq:conf_reg_pairs}
        \Pr[\norm{\hat{\varrho}_\pairset - \varrho_\pairset} \leq \eps \sigma_{\max}] \geq 1 - \delta
    \end{equation}
    for all pairs of qubits $\pairset$.

    In the following, we show that $\sigma_\pairset$ is related to the volume spanned by the measurement directions for the qubit pair $\pairset$.
    Recall that for the pair of qubits $\pairset = \{i,j\}$, the nine measurement directions are given by $\vec{v}_\alpha^{(i)} \otimes \vec{v}_\alpha^{(j)}$, $\alpha \in \{1,\dots,9\}$.
    We decompose the measurement map $M_\pairset$ of Eq.~\eqref{app-eq:meas_map} as $M_\pairset = {A}_\pairset B$.
    The matrix $B$ is such that its $i$th row, $i \in \{1,\dots,16\}$ is given by the vectorisation of the $i$th Pauli operator ordered as $(\eye \eye, X\eye, Y\eye, Z\eye, \eye X, \eye Y, \eye Z, XX, XY, XZ, YX, YY, YZ, ZX, ZY, ZZ)$.
    The matrix ${A}_\pairset$ is given by the $36 \times 16$ real matrix
    \begin{equation}
        A_\pairset = \frac{1}{36}
            \begin{pmatrix}
                1 & \vec{v}_1^{(i)} & \vec{v}_1^{(j)} & \vec{v}_1^{(i)} \otimes \vec{v}_1^{(j)} \\
                1 & \vec{v}_1^{(i)} & -\vec{v}_1^{(j)} & -\vec{v}_1^{(i)} \otimes \vec{v}_1^{(j)} \\
                1 & -\vec{v}_1^{(i)} & \vec{v}_1^{(j)} &- \vec{v}_1^{(i)} \otimes \vec{v}_1^{(j)} \\
                1 & -\vec{v}_1^{(i)} &- \vec{v}_1^{(j)} & \vec{v}_1^{(i)} \otimes \vec{v}_1^{(j)} \\
                1 & \vec{v}_2^{(i)} & \vec{v}_2^{(j)} & \vec{v}_2^{(i)} \otimes \vec{v}_2^{(j)} \\
                \vdots & \vdots & \vdots & \vdots \\
                1 & -\vec{v}_9^{(i)} & -\vec{v}_9^{(j)} & \vec{v}_9^{(i)} \otimes \vec{v}_9^{(j)} 
            \end{pmatrix}.
    \end{equation}
    We define the $16 \times 16$ matrix {$X_\pairset = M_\pairset^\dagger M_\pairset =  B^\dagger A_\pairset^T A_\pairset B$} and compute
    \begin{equation}
        A_\pairset^T A_\pairset = \frac{4}{36^2} 
            \begin{pmatrix}
                9 & 0 & 0 & 0 \\
                0 & \sum_{\alpha = 1}^9 \vec{v}_\alpha^{(i)}  (\vec{v}_\alpha^{(i)})^T & 0 & 0 \\
                0 & 0 & \sum_{\alpha = 1}^9 \vec{v}_\alpha^{(j)}  (\vec{v}_\alpha^{(j)})^T & 0 \\
                0 & 0 & 0 & \sum_{\alpha = 1}^9 (\vec{v}_\alpha^{(i)} \otimes \vec{v}_\alpha^{(j)})  (\vec{v}_\alpha^{(i)} \otimes \vec{v}_\alpha^{(j)})^T
            \end{pmatrix}.
    \end{equation}
    On the other hand, we can express $\sigma_\pairset$ as 
    \begin{equation}
        \sigma_\pairset = \max_{\vec{e}} \norm{M_\pairset^+ \vec{e}} = \max_{\vec{e}} \sqrt{\vec{e}^T (M_\pairset^+)^\dagger M_\pairset^+ \vec{e}},
    \end{equation}
    where the maximum is taken over vectors $\vec{e}$ from the standard basis, and where the norm is the euclidean norm. 
    The pseudoinverse is chosen to be the Moore-Penrose inverse, and from the singular value decomposition of the measurement map $M_\pairset = \sum_{i=1}^{16} \mu_i \vec{u}_i \vec{w}_i^\dagger$, we can write $M_\pairset^+ = \sum_{i=1}^{16} \nicefrac{1}{\mu_i} \vec{w}_i \vec{u}_i^\dagger$, such that 
    \begin{equation} \label{app-eq:sigma_ineq}
        \sigma_\pairset \leq \sqrt{\sum_{i=1}^{16} \frac{1}{\mu_i^2}} = \sqrt{\sum_{i=1}^{16} \frac{1}{\nu_i}},
    \end{equation}
    where $\nu_i$, with $i \in \{1,\dots,16\}$, are the eigenvalues of $X_\pairset$.
    We denote the last $9 \times 9$ block of $A_\pairset^T A_\pairset$ by $Y_\pairset$, that is,
    \begin{equation}
        Y_\pairset = \sum_{\alpha = 1}^9 (\vec{v}_\alpha^{(i)} \otimes \vec{v}_\alpha^{(j)}) (\vec{v}_\alpha^{(i)} \otimes \vec{v}_\alpha^{(j)})^T,
    \end{equation}
    and its eigenvalues are non-negative numbers $\lambda_i$, $i \in \{9\}$. The matrix 
    \begin{equation}
        Z_\pairset = (v_1^{(i)} \otimes v_1^{(j)}, v_2^{(i)} \otimes v_2^{(j)}, \dots, v_9^{(i)} \otimes v_9^{(j)})
    \end{equation}
    is such that $Y_\pairset = Z_\pairset Z_\pairset^T$, and therefore $Z_\pairset$ has singular values $\sqrt{\lambda_i}$, $i \in \{1,\dots,9\}$.
    Finally, we can write the determinant of $Z_\pairset$ as $\abs{\det(Z_\pairset)} = \prod_{i=1}^9 \sqrt{\lambda_i}$.
    By denoting the second and third block of $A_\pairset$ by $Y_i$ and $Y_j$ respectively, we can write $\det(X_\pairset) = \nicefrac{1}{36} \det(BB^\dagger) \det(Y_iY_j) \det(Z_\pairset)^2$.
    Due to Eq.~\eqref{app-eq:sigma_ineq}, we can expect that large $\abs{\det(Z_\pairset)}$ lead to small $\sigma_\pairset$.
    We note that $\abs{\det(Z_\pairset)}$ can be interpreted as the volume of the parallelotope spanned by the column vectors of $Z_\pairset$, and is maximal for columns that form an orthonormal basis.

    Taking a (product) orthonormal basis as measurement directions for the pair $\pairset = \{i,j\}$ directly leads to the maximum of $\abs{\det(Z_\pairset)} = 1$.
    One obvious choice is to take the standard basis, resulting in $Z_\pairset$ to be the nine-dimensional identity matrix.
    However, since the number of Pauli settings needed for two-body overlapping tomography of a six-qubit system is given by $\phi_2(6) = 12 > 9$, it is not possible to find $54$ local measurement directions $\vec{v}_\alpha^{(i)}$, with $\alpha \in \{1,\dots,9\}$ and $i \in \{1,\dots,6\}$, such that for every pair of qubits $\pairset = \{i,j\}$, the nine vectors $\vec{v}_\alpha^{(i)} \otimes \vec{v}_\alpha^{(j)}$, with $\alpha \in \{1,\dots,9\}$, form the standard basis.
    As a consequence, our goal is to find 54 local measurement directions $\vec{v}_\alpha^{(i)}$ ($\alpha \in \{1,\dots,9\}$ and $i \in \{1,\dots,6\}$) such that for each of the $15$ pairs $\pairset$, $\abs{\det(Z_\pairset)}$ is large.
    As a direct maximisation of $\min_\pairset \abs{\det(Z_\pairset)}$ is not easy, we turn our attention to an objective function of the type
    \begin{equation} \label{app-eq:obj_fct}
        f\left(\{\vec{v}_\alpha^{(i)}\}_{\alpha, i= 1}^{9, 6} \right) = w_1 \sum_\pairset \abs{\det(Z_\pairset)} - w_2 \sum_\pairset \abs{\det(Z_\pairset)}^2
    \end{equation}
    with weights $w_1, w_2 \geq 0$ and $w_1^2+w_2^2=1$.
    This is inspired by modern portfolio theory, or mean-variance analysis, which is a framework for assembling a collection of investments such that the expected return is maximised for a given level of risk in finance \cite{markowits1952portfolio}.
    The theory has initially been introduced by Harry Markowitz, for which he was eventually awarded the Nobel Memorial Prize in Economic Sciences in 1990.
    This allows us to find directions that will lead to large $\abs{\det(Z_\pairset)}$, while keeping the standard deviation of all $\abs{\det(Z_\pairset)}$ small.
    When the standard deviation is zero, clearly all $\abs{\det(Z_\pairset)}$ are equal.
    
    Using the Broyden-Fletcher-Goldfarb-Shanno (BFGS) algorithm, we maximised the objective function Eq.~\eqref{app-eq:obj_fct} for different weights, and found that for $w_2 \lesssim \cos(\pi/5)$, the achieved $\abs{\det(Z_\pairset)}$ are often equal for all pairs $\pairset$ (see Fig.~\ref{fig:modern_port_plot}).
    Using this approach, we were able to find $54$ measurement directions with $\sigma_{\max} \simeq 7.65$.
    We discuss in the next section how this reflects on the number of samples.
    \begin{figure}
        \centering
        \includegraphics[width=.75 \linewidth]{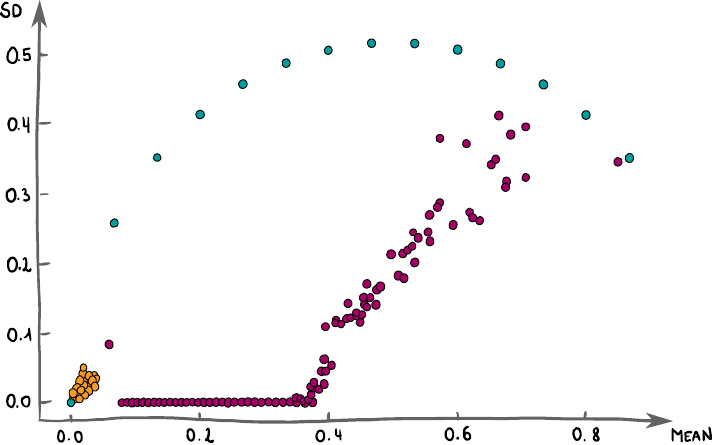}
        \caption{%
        Each point corresponds to one set of measurement directions, with its $x$-coordinate being the mean $\nicefrac{1}{15}\sum_\pairset \abs{\det(Z_\pairset)}$, and its $y$-coordinate being the standard deviation. 
        We aim to find measurement directions for which the $\abs{\det(Z_\pairset)}$ are equal for all pairs (i.e., standard deviation equal to zero) but also as large as possible.
        The \textcolor{orange}{orange} points correspond to randomly chosen vectors, the \textcolor{mypurple}{purple} points correspond to optimised directions with different weights (see Eq.~\eqref{app-eq:obj_fct}), and the \textcolor{myteal}{green} points correspond to cases where the vectors all are from the standard basis (when vectors are from the standard basis, the determinants of $Z_\pairset$ are either zero or one and correspond to Pauli settings).
        }
        \label{fig:modern_port_plot}
    \end{figure}

    Further, we may force the local directions to form three orthonormal bases, that is, for qubit $i$ we force the nine vectors $\vec{v}_\alpha^{(i)}$ ($\alpha \in \{1,\dots,9\}$) to be partitioned into three orthonormal bases.
    Doing this, we were able to obtain measurement directions with $\sigma_{\max} \simeq 7.78$. 
    The single-qubit $i \in \{1,\dots,6\}$ directions are partitioned in orthonormal bases following 
    \begin{equation*}
        \begin{tabular}{cccc}
            \{\{1, 2, 3\}, & \{4, 5, 6\}, & \{7, 8, 9\}\}, & \\
            \{\{1, 2, 5\}, & \{3, 4, 7\}, & \{6, 8, 9\}\}, & \\
            \{\{1, 2, 8\}, & \{3, 6, 7\}, & \{4, 5, 9\}\}, & \\
            \{\{1, 3, 4\}, & \{2, 8, 9\}, & \{5, 6, 7\}\}, & \\
            \{\{1, 5, 9\}, & \{2, 4, 7\}, & \{3, 6, 8\}\} & \text{and} \\
            \{\{1, 6, 7\}, & \{2, 4, 9\}, & \{3, 5, 8\}\} &
        \end{tabular}
    \end{equation*}
    respectively.
    The measurement directions are presented in Table \ref{app-tab:meas_dir}.
    \begin{table}
        \makebox[\textwidth][c]{
        \begin{tabular}{|c|c|c|c|c|c|}
            \hline
            Qubit 1  & Qubit 2  & Qubit 3  & Qubit 4  & Qubit 5 & Qubit 6 \\ \hline
            $[1.34851, -1.7187]$  & $[0.74451, 1.85896]$  & $[2.81234, -1.66384]$ & $[1.22444, -2.24737]$ & $[2.62025, -1.56922]$ & $[1.61654, 2.41608]$  \\
            $[1.62452, -0.16006]$ & $[0.83181, -1.13389]$ & $[1.24291, -1.56911]$ & $[1.86266, 1.89939]$  & $[0.78386, 2.4564]$   & $[0.32988, 0.14995]$  \\
            $[0.2289, 1.17782]$   & $[0.83405, -0.17509]$ & $[2.33714, 0.27682]$  & $[0.74332, -0.27308]$ & $[2.61964, -1.73379]$ & $[1.32478, -1.55164]$ \\
            $[0.88628, 0.06155]$  & $[0.98653, -2.38924]$ & $[2.63105, -1.5318]$  & $[0.94539, 2.20137]$  & $[1.04552, 0.26519]$  & $[1.65486, 1.47209]$  \\
            $[0.9695, 2.22663]$   & $[1.64509, 0.36903]$  & $[1.06042, -1.5596]$  & $[1.70326, -2.5176]$  & $[1.60308, 3.08691]$  & $[2.42079, -0.27072]$ \\
            $[1.01301, -2.04348]$ & $[1.83028, 0.04163]$  & $[2.21489, 2.65553]$  & $[2.12737, 2.11179]$  & $[1.05058, -1.644]$   & $[1.56836, -2.29642]$ \\
            $[2.70374, 0.28677]$  & $[2.08781, -1.20394]$ & $[1.16136, 1.41667]$  & $[2.56575, -0.74011]$ & $[2.09094, 1.49761]$  & $[0.04581, 2.36286]$  \\
            $[1.69042, -1.54368]$ & $[1.55645, -1.52536]$ & $[1.54184, 3.13342]$  & $[2.56242, -2.33567]$ & $[1.532, 3.04614]$    & $[0.91042, 0.21545]$  \\
            $[1.9898, 3.11515]$   & $[2.88169, -3.04215]$ & $[1.58265, -3.12376]$ & $[1.08649, -2.97173]$ & $[2.09094, 1.49761]$  & $[1.2526, 3.01513]$  \\ \hline
        \end{tabular}
        }
        \caption{%
        Measurement directions used for the six-qubit experimental implementation. 
        Each entry is a pair of Bloch vector angles $[\theta, \phi]$ in radians. 
        Each line corresponds to one six-qubit measurement setting $\mathcal{M}_\alpha$ as defined in Eq.~\eqref{app-eq:meas_settings}, with $\alpha = 1, \dots, 9$.}
        \label{app-tab:meas_dir}
    \end{table}

   Lastly, we discuss a possible analytical ansatz for the measurement directions, which relies on the structure of the $\mathbb{R}^9$. 
   This vector space contains a six-dimensional symmetric subspace, so one possible starting point could be to choose the first six measurement directions $\vec{v}_\alpha^{(i)} = \vec{v}_\alpha$ to be equal for every qubit.
   The set $\{\vec{v}_\alpha \otimes \vec{v}_\alpha\}_{\alpha = 1}^6$ thus spans the symmetric subspace and all parties perform the same local measurement at the same time.
   Then, it is left to find another three measurement directions for every party such that, together with the symmetric vectors, these $\vec{v}_\alpha^{(i)} \otimes \vec{v}_\alpha^{(j)}$ span the whole space $\mathbb{R}^9$ for every possible choice of $i$ and $j$.  
   However, the directions obtained using this ansatz together with a numerical optimisation over the remaining vectors do not lead to confidence regions as small as the unrestrained optimisation over all nine measurement directions.

\section{Discussion on the number of samples} \label{app-sec:nbr_samples}
    In this appendix, we compare our different measurement settings for two-body overlapping tomography of six qubits.
    For this, we want to compare the total number of samples $N$ required to reach a certain level of confidence.
    We recall the confidence region equations (Eq.~\eqref{app-eq:conf_reg_pairs}) from Ref.~\cite{degois2023userfriendly}
    \begin{equation}\label{eq:cr-sigma-max} 
        \Pr[\norm{\hat{\varrho}_\pairset - \varrho_\pairset} \leq \eps \sigma_{\max}] \geq 1 - \delta, \quad \forall \pairset.
    \end{equation}
    Therein, $\sigma_{\max}$ depends on the measurement settings, and $\eps = 3\sqrt{u}(\sqrt u+\sqrt{u+1})$, with $u = 2\log(8/\delta)/9 N$.

    When quantum state tomography of a two-qubit state is performed with the nine two-body Pauli settings, the uniform $\sigma$ is equal to five. 
    We denote this by $\sigma_{\mathrm{Pauli}} = 5$. 
    The $12$ minimal Pauli settings obtained through Eq.~\eqref{app-eq:mip} (see Sec.~\ref{app-sec:min_pauli_sets}) lead to a $\sigma_{\max} = 6.52 = \sigma_{\mathrm{bin. prog.}}$.
    The measurement settings optimised in the previous section achieve $\sigma_{\max} = 7.65 = \sigma_{\mathrm{unres. opt.}}$ and $\sigma_{\max} = 7.78 = \sigma_{\mathrm{orth. opt.}}$ in the case of unrestricted optimisation and orthonormal bases optimisation respectively.
    Finally, we want to compare with the 21 settings from Refs.~\cite{cotler2020quantum, bonetmonroig2020nearly, yang2023experimental}, which have $\sigma_{\max} = 10.7 =\sigma_{\mathrm{lit.}}$.

    We fix the radius $\eps \sigma_{\max}$ to 0.1, and report in the following table how many more samples are needed when comparing a scheme with a larger $\sigma_{\max}$ to one with a smaller $\sigma_{\max}$.
    \begin{equation} \label{tab:sigmas}
        \begin{tabular}{|c|c|c|c|c|c|}
            \hline
             & $\sigma_{\mathrm{Pauli}}$ & $\sigma_{\mathrm{bin. prog.}}$ & $\sigma_{\mathrm{unres. opt.}}$ & $\sigma_{\mathrm{orth. opt.}}$ & $\sigma_{\mathrm{lit.}}$ \\ \hline
            $\sigma_{\mathrm{Pauli}}$ & $\cdot$ & $70\%$ & $130\%$ & $140\%$ & $360\%$ \\
            $\sigma_{\mathrm{bin. prog.}}$ & $70\%$ & $\cdot$ & $38\%$ & $42\%$  & $170\%$ \\
            $\sigma_{\mathrm{unres. opt.}}$ & $130\%$ & $38\%$ & $\cdot$ & $3.4\%$ &  $95\%$ \\
            $\sigma_{\mathrm{orth. opt.}}$ & $140\%$ & $42\%$ & $3.4\%$ & $\cdot$ & $89\%$ \\
            $\sigma_{\mathrm{lit.}}$ & $360\%$ & $170\%$ & $95\%$ & $89\%$ & $\cdot$ \\
            \hline
        \end{tabular}
    \end{equation}
    First, we directly notice that the construction from Refs.~\cite{cotler2020quantum, bonetmonroig2020nearly} needs $170\%$ more samples when compared to the minimal Pauli set to achieve the same confidence level. 
    Second, Table \eqref{tab:sigmas} shows that there is little difference between the unrestricted optimised settings and the settings partitioned in three orthonormal bases per qubit ($3.4\%$ more for the settings partitioned in bases).
    Finally, it shows that requiring the minimal number of measurement settings (i.e., nine) for two-body overlapping tomography of six qubits comes at a cost of more measurement samples (of the order of $40\%$ more) to reach the same confidence level than the optimal Pauli settings (which require $12$ measurement settings).

    For the experimental demonstration of overlapping tomography discussed in the main text, we have used the minimal Pauli settings (leading to $\sigma_{\text{bin.prog}}$) and the optimised non-Pauli settings shown in Table \ref{app-tab:meas_dir} above (resulting in $\sigma_{\text{unres.opt}}$). A total of $9437$ ($8088$) counts were collected for the minimal Pauli (non-Pauli) case. From the data, it is possible to reconstruct the marginals $\hat{\varrho}_{\mathcal{S}}$ by simply applying the inverse measurement maps $M^+_{\mathcal{S}}$ to the frequencies of the counts $\vec{f}$ (Appendix~\ref{app-sec:conf_regs_and_num_opt}). The confidence regions would then ensue directly from Eq.~\eqref{eq:cr-sigma-max}, and guarantee that for any of the marginals, the true state $\varrho_{\mathcal{S}}$ is inside the ball $\norm{\hat{\varrho}_{\mathcal{S}} - \varrho_{\mathcal{S}}}_2 < \tilde{\epsilon}$ with high probability (say, $1\sigma \approx 0.682$) where $\tilde{\epsilon} \approx 0.17$ ($0.22$) for the minimal Pauli (non-Pauli) settings. Notice that the estimates $\hat{\varrho}_{\mathcal{S}}$ obtained in this way will naturally differ from the estimates obtained through the maximum likelihood estimator discussed in the main text and in Appendix~\ref{app-sec:exp}. Our choice of providing a detailed analysis of the latter is to facilitate comparison to previous experimental results due to it being a common choice in the literature. 
    
    As a closing remark, we note that it is always possible to find measurement settings such that for any pair $\{i,j\}$ of qubits, the nine two-body Pauli settings all appear exactly the same amount of times. 
    These Pauli sets correspond to \emph{orthogonal arrays} with three symbols (see Ref.~\cite{hedayat1999orthogonal} for an introduction to the topic and see Ref.~\cite{sloanetables} for tables of orthogonal arrays).
    For instance, for up to $n=7$ qubits, it is possible to find eighteen measurement settings such that for every pair of qubits, the nine two-body Pauli settings all appear twice. 
    For up to $13$ qubits, there exist $27$ measurement settings such that for every pair of qubits, the nine two-body Pauli settings all appear thrice. 
    This ensures that $\sigma$ is constant for all the pairs, and equal to five \cite{degois2023userfriendly}, but this approach clearly does not lead to minimal Pauli sets.
    Interestingly, making use of such repetitions in the Pauli settings can lead to a quadratic improvement in the number of samples, as discussed in Ref.~\cite{veltheim2024multiset}.
    Therein, the authors also present a greedy algorithm to obtain Pauli settings with given numbers of repetitions for each $k$-body Pauli operator, which achieves an asymptotic quadratic improvement for the cases they consider. 
    Finally, we note that taking all the $3^n$ Pauli settings leads to tomographically complete measurements for all $k$ subsets of qubits for any $k \leq n$, where each $k$-body Pauli appears exactly $3^{n-k}$ times.
    Thus, if the number of settings is not a major obstacle in the experimental implementation, such a Pauli set could in principle be considered such that for large $n$ and small $k$, measuring each setting only a few times would be enough samples.

\section{Detailed experimental results} \label{app-sec:exp}

    \subsection{Characteristic feature of experimental 6-qubit Dicke state} \label{app-sec:charact_dicke}
        As full state tomography is experimentally prohibitive, in this part we complement the experimental results of the measurements on the $Z$, $X$ and $Y$ bases to characterise the six-qubit experimental state. 
        In each measurements part (see
        {Fig.~3 of the main text}
        ), we projectively measure the polarisation of the photons either along $\ket{H}/ \ket{V}$, $\ket{D}/\ket{A}= \nicefrac{1}{\sqrt{2}}(|H\rangle\pm|V\rangle)$ or $\ket{L}/\ket{R}=\nicefrac{1}{\sqrt{2}}(|H\rangle\pm i|V\rangle)$, which are the eigenvectors of Pauli settings $Z$, $X$ and $Y$, respectively.
        
        The experimental results are presented in Fig.~\ref{fig:xyzbases}. The blue bars denote the normalised experimental probabilities and the pale grey bars denote the theoretical predictions of the ideal state $\ket{D_{6}^{(3)}}$. The probabilities are normalised by the total number of coincidence counts and the acquisition time for each measurement setting is two hours. The errors are deduced from Poissonian counting statistics. As shown in Fig.~\ref{fig:xyzbases} (a), the evident $20$ terms on the Pauli $Z$ bases are consistent with those expected for $\ket{D_{6}^{(3)}}$. However, there are also coincidence counts in $HHHHVV$, $HHVVVV$ and permutations thereof. This kind of noise originates from higher orders of the spontaneous parametric down-conversion (SPDC) process, in particular, from the eight-photon emission, where two of eight photons get lost due to the finite experimental detection efficiency. The remaining six photons will be registered as six-fold detector clicks for the noisy part as follows \cite{Dickestate2014}:
        \begin{equation}
            \varrho_{\rm noise}=\frac{4}{7} \varrho_{D_{6}^{(3)}}+\frac{3}{14}\left[\varrho_{D_{6}^{(2)}}+\varrho_{D_{6}^{(4)}}\right]
        \end{equation}
        where $\varrho_{D_{6}^{(j)}}=|D_{6}^{(j)}\rangle\langle D_{6}^{(j)}|$ with $|D_{6}^{(2)}\rangle=\nicefrac{1}{\sqrt{15}} \sum_{i} \mathcal{P}_{i}(|H H H H V V\rangle)$ and $|D_{6}^{(4)}\rangle=\nicefrac{1}{\sqrt{15}} \sum_{i} \mathcal{P}_{i}(|H H V V V V\rangle)$. Hence, the whole experimental quantum system can be  specified by the model $\varrho_{\rm{exp}}=p \varrho_{D_{6}^{(3)}}+(\nicefrac{1-p}{2})\varrho_{\rm noise}$. In the experiment, the parameter $p$ is determined by the power of the pump laser.
        
        The results of measurements on the Pauli $X$ and $Y$ bases are shown in Fig.~\ref{fig:xyzbases}~(b) and (c), respectively. The state  $\ket{D_{6}^{(3)}}$ can be transformed in these bases as follows \cite{Dickestate2009-weinfurter}:
        \begin{equation}
        \ket{D_{6}^{(3)}}=\sqrt{\frac{5}{8}}|\mathrm{GHZ}_{6}^{\mp}\rangle+ \sqrt{\frac{3}{16}}(|D_{6}^{(4)}\rangle \mp|D_{6}^{(2)}\rangle),
        \end{equation}    
        where $\quad\left|\mathrm{GHZ}_{N}^{\mp}\right\rangle=\nicefrac{1}{\sqrt{2}} \left(|0\rangle^{\otimes N} \mp|1\rangle^{\otimes N}\right)$, 0 denotes $\{D, L\}$, 1 denotes $\{A, R\}$. From Fig.~\ref{fig:xyzbases} (b) and (c), we observe the GHZ contribution as pronounced probabilities for the leftmost bars, $|DDDDDD\rangle$ or $|RRRRRR\rangle$, and rightmost bars, $|AAAAAA\rangle$ or $|LLLLLL\rangle$. This is directly related to the symmetry of 6-qubit Dicke state with three excitations $\ket{D_{6}^{(3)}}$, whereas the GHZ state manifests its two terms only in a single basis.
        
        \begin{figure}[h]
            \centering
            \includegraphics[width=0.7\linewidth]{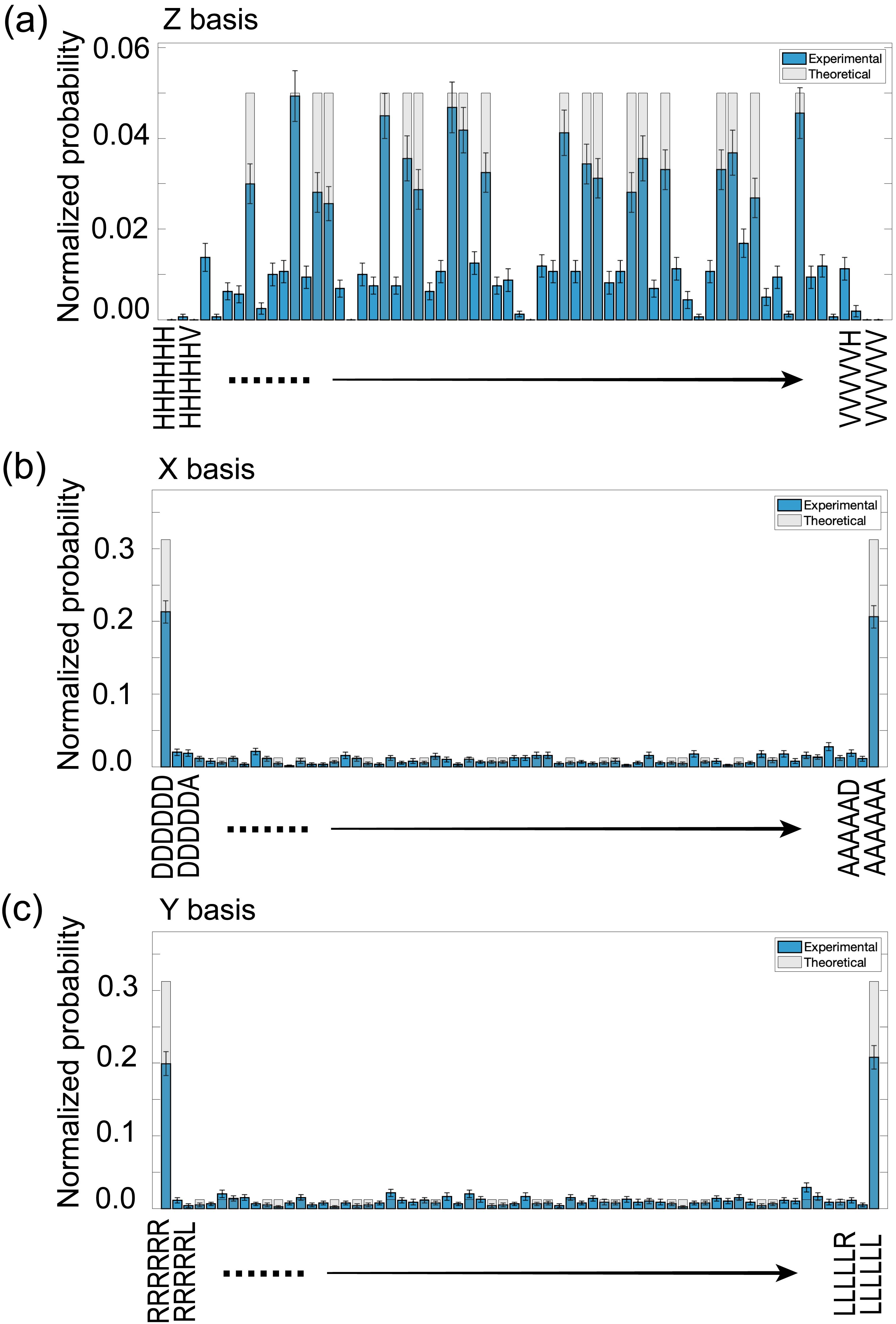}
            \caption{Experimentally measured normalised probability for the bases $Z$, $X$ and $Y$. The experimental results are denoted by the blue bars, which are normalised by the total number of coincidence counts. Theoretical predictions are shown as pale grey bars. The acquisition time for each measurement setting is two hours. The errors are deduced from Poissonian counting statistics.}
            \label{fig:xyzbases}
        \end{figure}

        \subsection{Reconstruction of the density matrices of all two-qubit subsystems} \label{app-sec:expe_reconstruct}
        In this part, we explain the method to reconstruct the density matrices of the two-qubit subsystems and provide the results of the experimental reconstructed density matrices. All the experimental density matrices are reconstructed by using maximum likelihood estimation (MLE) \cite{MLE2001}.
        
        A physical density matrix is Hermitian positive semidefinite, which is not guaranteed if the estimator is linear \cite{altepeter2005photonic}. This is due to the inherent statistical fluctuation in the number of counts in the experiment. 
        The idea of MLE is to find a physical density matrix $\varrho$ that is closely aligned with the observed experimental data. The set of tomographically complete observables is denoted by $\{\meas_\alpha\}_\alpha$. Each $\meas_\alpha$ can be decomposed into its measurement effects $\{\Pi_\alpha^o\}_o$, with $\{o\}$ labelling the possible measurement outcomes of $\meas_\alpha$. When a state $\varrho$ is measured according to $\mathcal{M}_\alpha$, the outcome $o$ occurs with probability $\tilde{p}_{o \mid \alpha}= \tr\left(\Pi^o_\alpha \varrho\right)$ and the number of counts is expected to $\tilde{N}_\alpha^o$. In the experiment, the measurement results consist of a set of counts $\{N_\alpha^o\}$ recorded for the $\alpha$th measurement setting and $o$th detector outcome combination. Assuming that the coincidence measurements has a Gaussian probability statistics, the probability $p(N_\alpha^o; \varrho)$ of obtaining the counts $N_\alpha^o$ is proportional to
        \begin{equation}
            p(N_\alpha^o; \varrho) \propto \exp[-\frac{(N_\alpha^o - \tilde{N}_\alpha^o)^2}{2 (\sigma_\alpha^o)^2}].       
        \end{equation}
        Thus, the likelihood that the matrix $\varrho$ would produce the measurement results $\{N_\alpha^o\}$ is        
        \begin{equation}
            p(\{N_\alpha^o\}; \varrho) = \frac{1}{\mathcal{N}} \prod_{\alpha,o} \exp[-\frac{(N_\alpha^o - N_\alpha\operatorname{tr}[\varrho \Pi_\alpha^o])^2}{2 (\sigma_\alpha^o)^2}],
            \label{eq:likelihood}
        \end{equation}
        where $\mathcal{N}$ is a normalisation constant, the standard derivation $\sigma_\alpha^o$ can be approximated as $\sigma_\alpha^o=\sqrt{\tilde{N}_\alpha^o} = \sqrt{{N_\alpha}\operatorname{tr}[\varrho \Pi_\alpha^o]}$, and $N_\alpha=\sum_{o}N_\alpha^o$ is the total counts for the observable $\meas_\alpha$.
        
        Then we want to find the state that maximises the likelihood of obtaining the counts $\{ N_\alpha^o \}$,
        \begin{equation}
            \varrho_{\text{MLE}} = \argmax_{\varrho \in S} p(\{N_\alpha^o\}; \varrho),
            \label{eq:mle}
        \end{equation}
        where $S$ denotes the set of physical density matrices. Considering a two-qubit quantum state $\varrho$, Eq.~\eqref{eq:mle} can be converted to an unconstrained optimisation problem by parameterising the state $\varrho$ as
        \begin{equation}
            \varrho(\vec{t}) = \frac{T(\vec{t})^\dag T(\vec{t})}{\operatorname{tr}[T(\vec{t})^\dag T(\vec{t})]},\label{eq:mle-unc}
        \end{equation}
        where
        \begin{equation}
            T(\vec{t}) = \begin{pmatrix}
                t_1 & 0 & 0 & 0 \\
                t_5 + \mathrm{i}t_6 & t_2 & 0 & 0\\
                t_7 + \mathrm{i}t_8 & t_9 + \mathrm{i}t_{10} & t_3 & 0\\
                t_{11} + \mathrm{i}t_{12} & t_{13} + \mathrm{i}t_{14} & t_{15} + \mathrm{i} t_{16} & t_4
            \end{pmatrix}
            \label{eq:t-param}
        \end{equation}
        and $\vec{t}=(t_1,t_2,\cdots,t_{16})$. In this way, while there are no constraints on $\vec{t}$ in Eq.~\eqref{eq:mle-unc}, the state $\varrho(\vec{t})$ is guaranteed to be physical \cite{MLE2001}. Then, the optimisation problem reduces to minimise the following cost function
        \begin{equation}
            \mathcal{L}(\vec{t})=\frac{1}{2} \sum_{\alpha,o} \frac{\left(N_\alpha^o-N_\alpha\operatorname{tr}\left[{\varrho}(\vec{t}){\Pi}_\alpha^o \right]\right)^{2}}{N_\alpha \operatorname{tr}\left[{\varrho}(\vec{t}){\Pi}_\alpha^o \right]}.
            \label{eq:likelihood2}
        \end{equation}

        Taking the two-qubit subsystem consisting of the first and the second qubits as an example, the projective measurement operators take the form
        \begin{equation}
            \Pi_{\alpha}^{(o_1,o_2)} = |\psi_\alpha^{o_1} \rangle\langle \psi_\alpha^{o_1}| \otimes |\varphi_\alpha^{o_2} \rangle\langle \varphi_\alpha^{o_2}| \otimes \mathbb{I} \otimes \mathbb{I} \otimes \mathbb{I} \otimes \mathbb{I}, \quad \alpha=1,\cdots,m \label{eq:meas-op}
        \end{equation}
        where $o_1,o_2=\{+,-\}$ denotes the possible outcomes on the two qubits and $|\psi_\alpha^-\rangle\langle \psi_\alpha^-| = \mathbb{I} - |\psi_\alpha^+\rangle\langle \psi_\alpha^+|$, $|\varphi_\alpha^-\rangle\langle \varphi_\alpha^-| = \mathbb{I} - |\varphi_\alpha^+\rangle\langle \varphi_\alpha^+|$ are the orthogonal projections, and $n=12,9$ respectively for the Pauli and non-Pauli overlapping tomography scheme. For each $\alpha$-th measurement setting, there are four measurement operators $|\psi_\alpha^\pm\rangle\langle \psi_\alpha^\pm|, |\varphi_\alpha^\pm\rangle\langle \varphi_\alpha^\pm|$, thus in total $4n$ two-qubit projective measurement operators. 
        
        The experimental number of counts for the projective measurement operator $\Pi_{\alpha}^{(o_1,o_2)}$ is calculated by
        \begin{equation}
            N_{\alpha}^{{(o_1,o_2)}} = \sum_{o_3,o_4,o_5,o_6=\pm} N_{\alpha}^{(o_1,o_2,o_3,o_4,o_5,o_6)}, \label{eq:raw-data}
        \end{equation}
        where $N_{\alpha}^{(o_1,o_2,o_3,o_4,o_5,o_6)}$ is the $2^6=64$ raw data obtained from six-fold coincidence measurements on the $\alpha$-th measurement setting with the $(o_1,o_2,o_3,o_4,o_5,o_6)$th detector combination clicked. Then, the normalised joint measurement probability distribution of the two-qubit subsystem can be calculated by
        \begin{equation}
            p_{\alpha}^{(o_1,o_2)}=\frac{N_{\alpha}^{(o_1,o_2)}}{\sum_{o_1,o_2=\pm} N_{\alpha}^{(o_1,o_2)}}.
            \label{eq:probability_distribution}
        \end{equation}    
        Substituting Eq.~\eqref{eq:meas-op} and Eq.~\eqref{eq:probability_distribution} into Eq.~\eqref{eq:likelihood2} yields the likelihood function that needs to be minimised to find the physical density matrix of the two-qubit subsystem,
        \begin{equation}
            \mathcal{L}=\frac{1}{2} \sum_{\alpha=1}^{m} \sum_{o_1,o_2=\pm}\frac{\left(p_{\alpha}^{(o_1,o_2)}-\operatorname{tr}\left[{\varrho}(\vec{t}){\Pi}_{\alpha}^{(o_1,o_2)}\right]\right)^{2}}{ \operatorname{tr}\left[{\varrho}(\vec{t}){\Pi}_{\alpha}^{(o_1,o_2)} \right]/[\sum_{o_1,o_2=\pm} N_{\alpha}^{(o_1,o_2)}]}.
            \label{eq:likelihood3}
        \end{equation}

        Finally, we present the experimental reconstructed density matrices for all two-qubit subsystems with the optimal overlapping tomography of Pauli measurements [Fig.~\ref{fig:densitymatrices}~(a)] and non-Pauli measurements [Fig.~\ref{fig:densitymatrices}~(b)]. 
        With the acquisition time 2 hours for each measurement setting, we collected total of 9437 (8088) counts for the minimal Pauli (non-Pauli) case. 
        The mutually overlaps of the two-qubit subsystems reconstructed by the Pauli scheme and the corresponding ones reconstructed by the non-Pauli scheme is shown from the average mixed state fidelity 0.963. {The prolonged and continuous measurement time for the entire experiment unavoidably introduced a slight drift in the system. In Fig.~4 of the main text, the error bars for the fidelities represent a $\pm1 \sigma$ uncertainty, corresponding to a $68\%$ confidence interval. However, if a $\pm2 \sigma$ uncertainty is considered, corresponding to $95\%$ confidence interval, there would be no regions outside the error bounds across the different measurement schemes.}
        \begin{figure}[h]
            \centering
            \includegraphics[height=0.9\textheight]{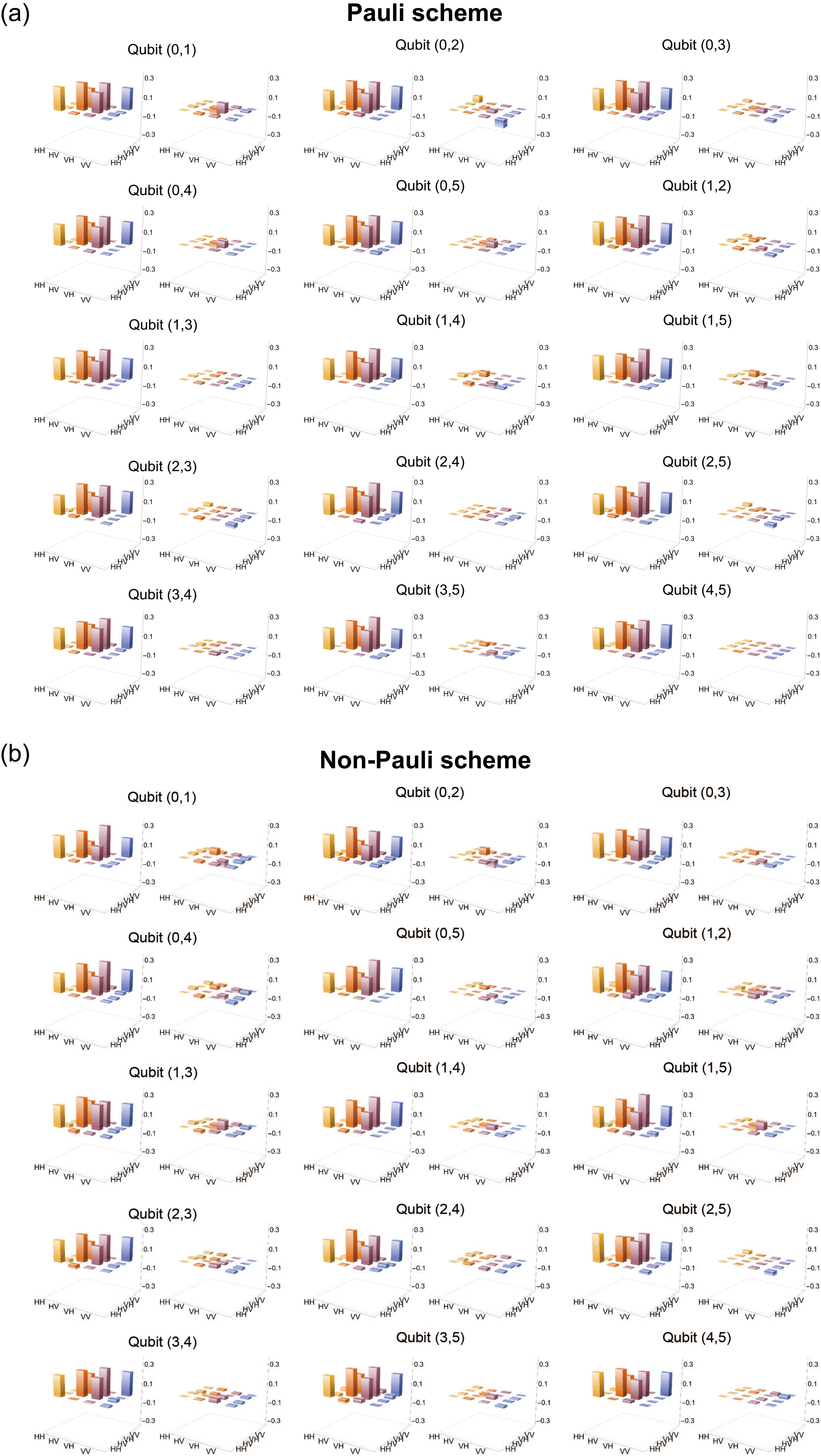}
            \caption{Experimental density matrices for all 15 two-body marginals reconstructed by the optimal overlapping tomography with (a) Pauli measurements and (b) non-Pauli measurements. For each pair of qubits, the left bars represent the real part while the right bars represent the imaginary part.}
            \label{fig:densitymatrices}
        \end{figure}

    \subsection{Discussion on the different error bars of experimental fidelities between Pauli and non-Pauli scheme} \label{app-sec:error_bars}
    Based on the experimental measurement data and the Monte Carlo simulations of 100 Poisson distribution samples, we obtained the average experimental fidelity for the fifteen two-body marginals $0.9091\pm0.0133$ for Pauli measurements and $0.8962\pm0.0180$ for non-Pauli measurements, both with the same acquisition time (2 hours) per measurement setting. 
    The difference in the average error bar is expected due to the different number of total counts in each experiment, but also due to the larger variances associated to the non-Pauli scheme. 
    Indeed, for non-Pauli scheme, a Monte Carlo simulation considering the same number of samples as the Pauli scheme ($9437$) leads to an average error bar of $0.0166$ ($1\sigma$ confidence level and $100$ trials) for the non-Pauli scheme, thus remains larger than that of the Pauli scheme.
    
    As estimated in Appendix~\ref{app-sec:nbr_samples}, requiring the minimal number of 9 measurement settings for two-body overlapping tomography of six qubits incurs a trade-off in the form of approximately $40\%$ more measurement samples to achieve the same confidence level as the optimal Pauli scheme, which require $12$ measurement settings. Considering a Monte Carlo simulation with 1.4 times total number of samples ($13212$) compared to the Pauli scheme, the average error bar for non-Pauli scheme is estimated to be $0.0134$ ($1\sigma$ confidence level and $100$ trials). 
    This value closely approaches the average error bar ($0.0133$) for Pauli scheme, which further validated our theoretical estimation of the difference of the sample cost for the two schemes.

\end{appendix}

\end{document}